\documentclass[conference]{IEEEtran}

\usepackage{graphics,
           psfrag,
           epsfig,
           amsthm,
           cite,
           amssymb,
           url,
           dsfont,
           subfigure,
           algorithm,
           algorithmic,
           balance,
           enumerate,
           color,
           setspace
}
\usepackage{amsmath}%[centertags][tbtags][sumlimits][nosumlimits][intlimits][namelimits][nonamelimits]

\usepackage{epstopdf}

\newtheorem{definition}{Definition}

\newtheorem{theorem}{Theorem}

\newcommand{\eref}[1]{(\ref{#1})}
\newcommand{\sref}[1]{Section~\ref{#1}}

\newcommand{\fref}[1]{Figure~\ref{#1}}

\newcommand{\cref}[1]{Constraint~\ref{#1}}

\newcommand{\ignore}[1]{}

\IEEEoverridecommandlockouts
\usepackage{slashbox}

\begin{document}

\title{Decoding Delay Controlled Reduction of Completion Time in Instantly Decodable Network Coding}

\author{
   \authorblockN{Ahmed Douik, \textit{Student Member, IEEE}, Sameh Sorour, \textit{Member, IEEE},\\ Tareq Y. Al-Naffouri, \textit{Member, IEEE}, and Mohamed-Slim Alouini, \textit{Fellow, IEEE}}
   {\thanks {Ahmed Douik and Mohamed-Slim Alouini are with Computer, Electrical and Mathematical Sciences and Engineering (CEMSE) Division at King Abdullah University of Science and Technology (KAUST), Thuwal, Makkah Province, Saudi Arabia, email: \{ahmed.douik,slim.alouini\}@kaust.edu.sa

Sameh Sorour is with the Electrical Engineering Department, King Fahd University of Petroleum and Minerals (KFUPM), Dhahran, Eastern Province, Saudi Arabia, email:samehsorour@kfupm.edu.sa

Tareq Y. Al-Naffouri is with both the CEMSE Division at King Abdullah University of Science and Technology (KAUST), Thuwal, Makkah Province, Saudi Arabia, and the Electrial Engineering Department at King Fahd University of Petroleum and Minerals (KFUPM), Dhahran, Eastern Province, Saudi Arabia, e-mail: tareq.alnaffouri@kaust.edu.sa.

This work is an extended version of work \cite{confarxiv} submitted to Globecom, Austin, Texas, USA, 2014.
}}
    }

\maketitle

\IEEEoverridecommandlockouts

\begin{abstract}
For several years, the completion time and the decoding delay problems in Instantly Decodable Network Coding (IDNC) were considered separately and were thought to completely act against each other. Recently, some works aimed to balance the effects of these two important IDNC metrics but none of them studied a further optimization of one by controlling the other. In this paper, we study the effect of controlling the decoding delay to reduce the completion time below its currently best known solution in persistent erasure channels. We first derive the decoding-delay-dependent expressions of the users' and overall completion times. Although using such expressions to find the optimal overall completion time is NP-hard, we design two novel heuristics that minimizes the probability of increasing the maximum of these decoding-delay-dependent completion time expressions after each transmission through a layered control of their decoding delays. We, then, extend our study to the limited feedback scenario. Simulation results show that our new algorithms achieves both a lower mean completion time and mean decoding delay compared to the best known heuristic for completion time reduction. The gap in performance becomes significant for harsh erasure scenarios.
\end{abstract}


\section{Introduction} \label{sec:intro}

\emph{Network Coding (NC)} gained much attention in the past decade after its first introduction in the seminal paper \cite{850663}. In the last lustrum, an important subclass network coding, namely the Instantly Decodable Network Coding (IDNC) was an intensive subject of research \cite{ref2,6655395,letterarxiv,ref3,6766433,ref4,xiao1,
ref5,6725590,ref6,6120247,ref8,xiao2,6620795,ref17,ref18,arg2,refsameh,5753573,refahmed,refjournal} thanks to its several benefits, such as the use of simple binary XOR to encode/decode packets, no buffer requirement, and fast progressive decoding of packets. which is much favorable in many applications (e.g. roadside to vehicle safety messages, satellite networks and IPTV) compared to the long buffering time needed in other NC approaches before decoding.

For as long as the research on IDNC has existed, there were two main metrics that were considered in the literature as measures of its quality, namely the \emph{completion time} \cite{ref4} and the \emph{decoding delay} \cite{ref2}. The former measures how fast the sender can complete the delivery/recovery of requested packets whereas the latter measures how far the sender is from being able to serve all the unsatisfied receivers in each and every transmission. For quite some time, these two metrics were considered for optimization separately in many works. Though both proved to be NP-hard parameters to minimize, many heuristics has been developed to solve them in many scenarios \cite{ref4,ref2,ref3,ref5,refjournal,refsameh,refahmed}, but again separately. In fact, it can be easily inferred from \cite{ref4} and \cite{ref2} that the policies derived so far to optimize one usually degrades the other.

It was not until very recently until one work has aimed to derive a policy that can balance between these two metric and achieve an intermediate performance for both of them \cite{nada}. Nonetheless, there does not exist, to the best of our knowledge, any work that aims to explore how these two metrics can be controlled together in order to achieve an even better performance than the currently best known solutions. For instance, every time an unsatisfied user receives a coded packet that is not targeting it, its decoding delay increases and so does its individual completion time. Although this fact was noted for erasure-less transmissions in \cite{nada}, it was used to strike a balance in performance between both metrics and not to investigate whether a smart control of such decoding delay effects will further reduce the overall completion time compared to its current best achievable performance.

In wireless networks, packet loss occurs due to many phenomena related to the mobility and the propagation environment and they are seen as packet erasures at higher communication layers \cite{1208721}. This erasure nature of the links affects the delivery of meaningful data and thus affects the ability of users to synchronously decode the information flow. As a consequence a better use of the channel and network does not mean an effective better throughput at higher communication layers \cite{1208721}. Numerous research has been dedicated to understand the different delay aspects in NC. In our previous work \cite{confarxiv}, we considered a prompt and perfect reception of the feedback. This assumption is not realistic in practice because of the impairments in the feedback channel. In this paper, we aim to extend the work in \cite{confarxiv} by studying the completion time reduction problem of G-IDNC in the persistent erasure channel (PEC) model on both the forward and backward channels, and in the presence of lossy feedback intermittence.

In this paper, we aim to design a new completion time reduction algorithm through decoding delay control in persistent erasure channels. We first derive more general expressions of the individual and overall completion times over erasure channels as a function of the users' decoding delays. Since finding the optimal schedule of coded packets to minimize the overall completion time is NP-hard \cite{arg1}, we design two greedy heuristic that aims to minimize the probability of increasing the maximum of these decoding-delay-dependent completion time expressions after each transmission. In the first heuristic, we show that this process can be done by partitioning the \emph{IDNC graph} into layers with descending order of user completion time criticality before each transmission. The coding combination for this transmission is then designed by going through these descending order layers sequentially and selecting the combination that minimizes the probability of any decoding delay increments within each layer. This is done while maintaining the instant decodability constraint of the overall coding combination for the targeted users in the more critical layer(s). In the second heuristic, we use a binary optimization algorithm with multi-layer objective function. We, then, extend our study to the limited feedback environment. Finally, we compare through simulations the performance of our designed algorithm to the best known completion time and decoding delay reduction algorithms.

The rest of this paper is divided as follows: In \sref{sec:sys}, we present our system model and in \sref{sec:ch} we present the channel and feedback model. The problem formulation using the decoding delay is presented in \sref{sec:formulation}. Analytic formulation of the sub optimal solution at each transmission is provided in \sref{sec:comp}. In \sref{sec:alg}, algorithms to solve the former problem are presented. \sref{sec:ext} presents the extension of the study to the limited feedback scenario. Before concluding in \sref{sec:conclusion}, simulations results are illustrated in \sref{sec:sim}.

\section{System Model and Parameters} \label{sec:sys}

The model we consider in this paper consists of a wireless sender that is required to deliver a frame (denoted by $\mathcal{N}$) of $N$ source packets to a set (denoted by $\mathcal{M}$) of $M$ users. Each user is interested in receiving the $N$ packets of $\mathcal{N}$. In an \emph{initial phase}, The sender transmits the $N$ packets of the frame uncoded. Each user listens to all transmitted packets and feeds back to the sender an acknowledgement for each successfully received packet.

After the \emph{initial phase}, two sets of packets are attributed to each user $i$ at the sender:
\begin{itemize}
\item The \emph{Has} set (denoted by $\mathcal{H}_i(t)$) is defined as the set of packets successfully received by user $i$.
\item The \emph{Wants} set (denoted by $\mathcal{W}_i(t)$) is defined as the set of packets that are lost by user $i$. In other words,we have $\mathcal{W}_i = \mathcal{N} \setminus \mathcal{H}_i$.
\end{itemize}

The BS saves the information obtained after the transmission at time $(t-1)$ in a \emph{feedback matrix} (FM) $\mathbf{F}(t) = [f_{ij}(t)],~ \forall~ i \in \mathcal{M},~ \forall~j \in \mathcal{N},~ \forall~t > 0$ such that:
\begin{align}
f_{ij}(t) =
\begin{cases}
0 \hspace{0.9 cm}& \text{if } j \in \mathcal{H}_i(t) \\
1 \hspace{0.9 cm}& \text{if } j \in \mathcal{W}_i(t).
\end{cases}
\end{align}

For ease of notation, we will assume that the time index $t$ denotes the transmission number within the recovery phase and thus $t=0$ refers to its beginning. Therefore, the sets $\mathcal{H}_i(0), \mathcal{W}_i(0)$ and $\mathcal{U}_i(0)$ refers to the sets at the beginning of the recovery phase (i.e. the sets obtained after the initial transmissions). After this initial transmission, the recovery phase starts at time $t=1$. In this phase, the BS uses binary XOR to encode the source packets to be send. The packet combination is chosen using the information available in the FM and the expected erasure patterns of the links. In this phase, the transmitted coded packets can be one of the following three options for each user $i$:

\begin{itemize}
\item \emph{Non-innovative:} A packet is non-innovative for user $i$ if it does not bring new useful information. In other words, all the source packets encoded in it were successfully received and decoded previously.
\item \emph{Instantly Decodable:} A packet is instantly decodable for user $i$ if the encoded packet contain at most one packet the user do not have so far. In other words, it contains \emph{only one packet} from $\mathcal{W}_i$.
\item \emph{Non-Instantly Decodable:} A packet is non instantly decodable for user $i$ if it contains more than one source packet the user do not have so far. In other words, it contains at least two packets from $\mathcal{W}_i$.
\end{itemize}

After the \emph{initial phase}, the \emph{recovery phase} begins. In this phase, the sender exploits the diversity in received packets at the different users to transmit network coded combinations of the source packets. After each transmission, users update the sender in case they receive the coded packet and decode missing source packets from it. This process is repeated until all users complete the reception of all frame packets. 

We define, like in \cite{letterarxiv}, the targeted users by a coded packet (or a transmission) as the users to whom the BS indented the packet combination when encoding it. Given a schedule $S$ of coded packets transmitted by the sender, we define the individual completion time, overall completion time and the decoding delay, like in \cite{refsameh,ref2,ref5,ref6,nada}, as follows:
\begin{definition}
The individual completion time $\mathcal{S}_i(S)$ of user $i$ is the number of recovery transmissions required until this user obtained all its requested packets.
\end{definition}
\begin{definition}
The overall completion time $\mathcal{S}(S)$ of a frame is the number of recovery transmissions required until all users obtain all their requested packets. It easy to infer that $\mathcal{S}(S) = \max_{i\in\mathcal{M}} \mathcal{S}_i(S)$.
\end{definition}
\begin{definition}
At any recovery phase transmission a time $t$, a user $i$, with non-empty Wants set, experiences a one unit increase of decoding delay if it successfully receives a packet that is either non-innovative or non-instantly decodable. Consequently, the decoding delay $D_i(S)$ experienced by user $i$ given a schedule $S$ is the number of received coded packets by $i$ before its individual completion, which are non-innovative or non-instantly decodable.
\end{definition}

\section{Forward Channel Model} \label{sec:ch}

Following the same model as \cite{ref3,ref17,refjournal}, the forward channel (from the BS to the users) is modeled as a Gilbert-Elliott channel (GEC) \cite{4607216}. The good and bad states are designed by G and B, respectively. In the original version of the Gilbert-Elliott channel \cite{BLTJ:BLTJ955}, the good state was assumed to be error free and the bad state always in error (i.e. $p=0$ and $q=1$ in \fref{fig:GEC}). This zero error probability in the good state allowed the computation in close form of the capacity \cite{45284}. This model was then extended to incorporate non zero error probability in the good state. This formulation allow the modeling of multiple fading scenarios \cite{4607216}. In this paper, this extended channel form will be used (see \fref{fig:GEC}).
\begin{figure}[t]
\centering
  \includegraphics[width=1\linewidth]{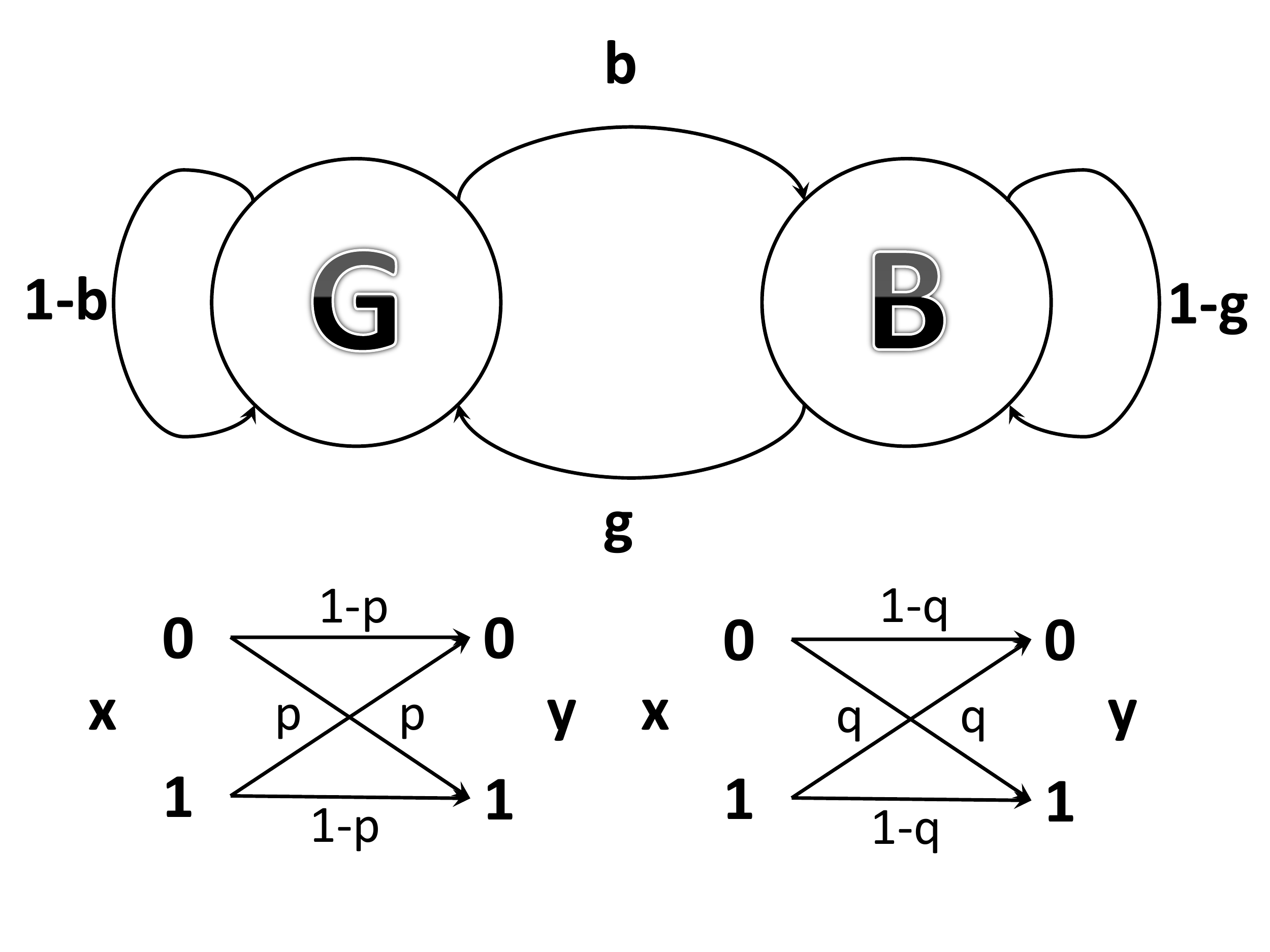}\\
  \caption{The two state Gilbert-Elliott channel.}\label{fig:GEC}
\end{figure}  

Due to the persistent erasure of the Markov chain, the channel has memory that depends on the transition probabilities between the good and bad states. Let $C_i^f$ denote the state of the forward channels. For each user $i$, the transition probability, from time $t-1$ to $t$, are:
\begin{align}
&\mathds{P} (C_i^f(t) = G | C_i^f(t-1) = B) = g_i^f \nonumber \\
&\mathds{P} (C_i^f(t) = B | C_i^f(t-1) = B) = 1-g_i^f \nonumber \\
&\mathds{P} (C_i^f(t) = B | C_i^f(t-1) = G) = b_i^f \nonumber \\
&\mathds{P} (C_i^f(t) = G | C_i^f(t-1) = G) = 1-b_i^f,
\end{align}
where the superscript $f$ indicates the transmission (forward) channel. The use of a superscript ($f$ herein) will be motivated in \sref{sec:ext} when studying the limited feedback scenario. The probabilities, for user $i$, for a packet to be erased in the good and bad state are respectively $p_i^f$ and $q_i^f$.

Since this Markov chain is time-homogeneous (the process can be described by a single, time-independent matrix), we express the probabilities to be in Good or Bad state (steady-state probabilities) for both the transmission and feedback channel, as:
\begin{align}
&\mathcal{P}_{G_i^f} = \mathds{P} (C_i^f = G) = \cfrac{g_i^f}{g_i^f+b_i^f} \nonumber \\ 
&\mathcal{P}_{B_i^f} = \mathds{P} (C_i^f = B) = \cfrac{b_i^f}{g_i^f+b_i^f},
\end{align}

We define the memory factor as an indicator of the correlation between the states at different times. A high value of this factor means that the channel is likely to stay in the same state during the following transmissions. In contrast, a small value indicates that the state of the channel changes in an independent manner. For each user $i$, the memory factor of the forward channel $\mu_i$ can be formulated as:
\begin{align}
\mu_i = 1 - g_i^f - b_i^f 
\end{align}

The persistent erasure channel is more likely to stay in the state during the next transmission than switching states. Therefore, we have $0 \leq \mu_i \leq 1,~ \forall~i \in \mathcal{M}$. Let $\mu = \cfrac{\sum\limits_{i \in \mathcal{M}} \mu_i }{ M}$ be the average memory for the forward link. The memory-less channels are a special case of this persistent erasure channel that can be obtained by setting the memory factor for each user to $0$. In other words, by setting $\mu=0$.

We assume that each user is seen through a channel that is independent from all the other users and thus no correlation exists between the different channels. The state transition probabilities of all the users are know by the BS and used when choosing the packet to be encoded. 

\section{Problem Formulation using Decoding-Delay-Dependent Expressions} \label{sec:formulation}

Let $d_i(\kappa(t))$ be the decoding delay increase for user $i$, at time $t$, after the transmission $\kappa(t)$. Define $D_i(t)$ as the total decoding delay experienced by user $i$ until the transmission at time $t$ (i.e. $D_i(t) = \sum\limits_{n=1}^t d_i(\kappa(t))$). Let $\mathcal{C}_i(S)$ denote the completion time for user $i$ given a certain schedule $S$ of coded packets. In other words, $\mathcal{C}_i$ is the total number of transmissions required, since the beginning of the \emph{recovery phase}, so that user $i$ successfully receives all its missing packets.

The completion time for the whole session, denoted by $\mathcal{C}$, is the required time, from the beginning of the \emph{recovery phase}, to serve all users. Therefore, $\mathcal{C}$ is controlled by the largest user completion time. In other words, we have:
\begin{align}
\mathcal{C} = \underset{i \in \mathcal{M}}{\text{max }} \mathcal{C}_i.
\end{align}

The following theorem introduces a decoding-delay-dependent expression for the individual completion time of user $i$ and the overall completion time, given the transmission of schedule $S$ from the sender over erasure channels.
\begin{theorem}
For a relatively large number of packets $N$, and a schedule $S$ of transmitted packets transmitted by the sender until overall completion occurs to all users, the individual completion time for user $i$ can be approximated by:
\begin{equation}\label{eq:ICT}
\mathcal{C}_i(S) \approx \frac{\left|\mathcal{W}_i(0)\right| + D_i(S)-\alpha_i}{1-\alpha_i}
\end{equation}
where
\begin{align}
\alpha_i = \cfrac{g_i^fp_i^f+q_i^fb_i^f}{g_i^f+b_i^f}
\end{align}
Consequently, the overall completion time for the same schedule $S$ can be expressed as:
\begin{equation}
\mathcal{C}(S) \approx \max_{i\in\mathcal{M}}\left\{\frac{\left|\mathcal{W}_i(0)\right| + D_i(S)-\alpha_i}{1-\alpha_i}\right\}
\end{equation}
\end{theorem}
\begin{proof}
The proof can be found in Appendix B.
\end{proof}
In the rest of the paper, we will use the approximation with equality as it indeed holds for large $N$.
We can thus formulate the minimum completion time problem as finding the schedule of coded packet $S^*$, such that:
\begin{align}\label{eq:opt}
S^* &= \arg\min_{S\in\mathcal{S}} \left\{\mathcal{C}(S)\right\} \nonumber \\ &= \arg\min_{S\in\mathcal{S}}\left\{ \max_{i\in\mathcal{M}}\left\{ \frac{|\mathcal{W}_i(0)| + D_i(S)-\alpha_i}{1-\alpha_i} \right\}\right\}\;,
\end{align}
where $\mathcal{S}$ is the set of all possible transmission schedules of coded packets.

Clearly, finding this optimal schedule at time $t=0$ through the above optimization formulation is very difficult. This is true due to the dynamic nature of erasures and the dependence of the optimal schedule of their effect, which makes the above equations anti-causal (i.e. current result depends on input from the future). Moreover, we know from the literature that optimizing the completion time over the whole \emph{recovery phase} is intractable \cite{refsameh} even for the erasure-less scenario \cite{nada}. On the other hand, this formulation shows that the only terms affected by the schedule in the individual and overall completion time expressions are the decoding delay terms of the different users. Consequently, controlling such decoding delays in a smart way throughout the selection of the coded packet schedule can indeed affect the reduction of the completion time significantly. We will thus design a new heuristic algorithm in the next section that takes this fact into consideration. In the rest on this paper, we will refer to the best packet combination that reduce the completion time at each time step as the optimal solution.

\section{Completion Time Reduction}\label{sec:comp}

\subsection{Critical Criterion}

From \eref{eq:opt}, we can see that the optimal schedule is the one that achieves the minimum overall growth in the individual completion time expressions in \eref{eq:ICT} $\forall~i\in\mathcal{M}$. Since we know that finding such schedule for the entire recovery phase, prior to its start, is intractable, we will design our heuristic algorithm such that, in each transmission a time $t>0$, it minimizes the probability of increase of the maximum of such expressions over all users compared to their state before this transmission. To formally express this criterion, let us first define $\mathcal{C}_i(t)$ as:
\begin{equation} \label{eq:Ct}
\mathcal{C}_i(t) = \frac{|\mathcal{W}_i(0)| + D_i(t) - \alpha_i}{1-\alpha_i}
\end{equation}
In other words, $C_i(t)$ is the anticipated individual completion time of user $i$ if it experiences no further decoding delay increments starting from time $t$. Thus, the philosophy of our proposed heuristic algorithm is to transmit the coded packet $\kappa(t)$ at time $t$ such that:
\begin{align}\label{eq:heuristic-criterion}
\kappa^{*}&(t) \nonumber \\
&= \arg\min_{\kappa(t)\in\mathcal{G}(t)} \left\{\mathds{P}\left(\max_{i\in\mathcal{M}} \left\{\mathcal{C}_i(t)\right\} > \max_{i\in\mathcal{M}} \left\{\mathcal{C}_i(t-1)\right\}\right) \right\}
\end{align}
We will refer to \eref{eq:heuristic-criterion} as the critical criterion. Let $\mathcal{P}(t)$ be the set of users that can potentially increase $\max_{i\in\mathcal{M}} \left\{\mathcal{C}_i(t)\right\}$ at time $t$ compared to $\max_{i\in\mathcal{M}} \left\{\mathcal{C}_i(t-1)\right\}$ if they are not targeted by $\kappa(t)$. The set can be mathematically defined as follows:
\begin{align}\label{eq:critical-set}
\mathcal{P}(t) = \Bigg\{i\in\mathcal{M} \;\Biggm|\; &\frac{|\mathcal{W}_i(0)| + \left(D_i(t-1)+1\right) - \alpha_i}{1-\alpha_i} \nonumber\\
 & \qquad >  \frac{|\mathcal{W}_j(0)| + D_j(t-1)-\alpha_j}{1-\alpha_j}\Bigg\}\;,
\end{align}
where $j =   \arg\max_{k \in \mathcal{M}} \left\{ \frac{|\mathcal{W}_k(0)| + D_k(t-1) - \alpha_k}{1-\alpha_k}\right\}$. We will refer to this set as the "highly critical set".

\subsection{Optimization Problem}

Let $e_i(t)$ be the probability that user $i$ loose the transmission at time $t$. This formulation of this probability is:
\begin{align}
e_i(t) = 
\begin{cases}
p_i \hspace{2cm} \text{ if } C_i^f(t)=G\\
q_i \hspace{2cm} \text{ otherwise } 
\end{cases}
\end{align}

Define $\tau(\kappa(t))$ as the set of users that are targeted by the transmission $\kappa(t)$. The following theorem defines a maximum weight clique algorithm that can satisfy the critical criterion.
\begin{theorem}\label{th:critical-criterion}
The critical criterion in \eref{eq:heuristic-criterion} can be achieved by selecting $\kappa^*(t)$ according to the following optimization problem:
\begin{align}\label{eq:criterion-optimization}
&\kappa^{*}(t) = \underset{\kappa(t) \in \mathcal{G}}{\text{argmax}} \left\{ \sum_{i \in (\mathcal{P}(t) \cap \tau ) } \text{log}\left( \cfrac{1}{ e_i(t)}\right) \right\}.
\end{align}
In other words, the transmission $\kappa(t)$ than can satisfy the critical criterion can be selected using a maximum weight clique problem in which the weight of each vertex $v_{ij}$ in $\mathcal{P}(t)$ can be expressed as:
\begin{align}\label{eq:weights}
w_{ij}^* = \text{log}\left( \cfrac{1}{ e_i(t)}\right).
\end{align}
\end{theorem}
\begin{proof}
The proof can be found in \cite{confarxiv}.
\end{proof}

\section{Proposed Heuristics Algorithms}\label{sec:alg}

\subsection{Maximum Weight Clique Solution}

\subsubsection{Instantly Decodable Network Coding Graph}

To look for possible packet combinations the base station can make, we use the G-IDNC graph representation. To construct this G-IDNC graph $\mathcal{G} (\mathcal{V},\mathcal{E})$, we first create a vertex $v_{ij} \in \mathcal{V}$ for each packet $j \notin \mathcal{H}_i,~ \forall~ i \in \mathcal{M}$. We then connect with an edge all $v_{ij}$ and $v_{kl}$ in $\mathcal{V}$ if one of the two following conditions is true:
\begin{itemize}
\item $j=l \Rightarrow$ Packet $j$ is needed by both clients $i$ and $k$.
\item $j \in \mathcal{H}_k$ and $l \in \mathcal{H}_i \Rightarrow$ The packet combination $j\oplus l$ is instantly decodable for both clients $i$ and $k$.
\end{itemize}

Given this graph formulation and according to the analysis done in \cite{arg1}, the set of all packet feasible combinations in G-IDNC is represented by all maximal cliques in $\mathcal{G}$. To generate a packet combination, binary XOR is applied to all the packets identified by the vertices of a selected maximal clique $\kappa$ in $\mathcal{G}$. The targeted clients by this transmission $\kappa$ are those identified by the vertices of the selected maximal clique.

\subsubsection{Multi-layer Greedy Algorithm}

Despite the importance of the satisfaction of the critical criterion in order to minimize the probability of increase of the maximum individual completion time, it may not fully exploit the power of IDNC. In other words, once a clique is chosen according to \eref{eq:criterion-optimization} from among the users in the highly critical set $\mathcal{P}(t)$, there may exist vertices belonging to other users that can form an even bigger clique. Thus, adding this vertex to the clique and serving this user will benefit them without affecting the IDNC constraint for the users belonging to $\mathcal{P}(t)$.  

To schedule such vertices and their users, let $\mathcal{G}_1,\mathcal{G}_2,...\mathcal{G}_h$ (with $h \in \mathds{N}$) be the sets of vertices of $\mathcal{G}(t)$, such that $v_{ik} \in \mathcal{G}_n$ if the following conditions are true:
\begin{itemize}
\item $\mathcal{C}_i(t-1)+ \cfrac{n}{1-\alpha_i} > \mathcal{C}_j(t-1)$.
\item $\mathcal{C}_i(t-1)+ \cfrac{n-1}{1-\alpha_i} \leq \mathcal{C}_j(t-1)$.
\end{itemize}
where $j =   \underset{ i \in \mathcal{M}}{\text{argmax }} \left\{ \mathcal{C}_i(t-1)\right\}$. Consequently, the IDNC graph at time $t$ is partitioned into $h$ layers with descending order of criticality. By examining the above condition, the vertices of the users of $\mathcal{P}(t)$ are all in layer $\mathcal{G}_1$. Moreover, the $n$-th layer of the graph includes the vertices of the users who may eventually increase $\max_{i\in\mathcal{M}}\left\{\mathcal{C}_i(t+n)\right\}$ if they experience $n$ decoding delay increments in the subsequent $n$ transmissions. Consequently, a user with vertices belonging to $\mathcal{G}_i$ is more critical than another with vertices belonging to $\mathcal{G}_j$, $j>i$, as the former has a higher chance to increase the overall completion time.

In order to guarantee the satisfaction of the critical criterion, our proposed algorithm first finds the maximum weight clique $\kappa^*$ in layer $\mathcal{G}_1$ as mandated by Theorem \ref{th:critical-criterion}. We then construct $\mathcal{G}_2(\kappa^*)$ including each vertex in $\mathcal{G}_2$ that is adjacent to all vertices in $\kappa^*$ (i.e. forms a bigger clique with $\kappa^*$). After assigning the same weights defined in \eref{eq:weights}, the maximal weight clique in $\mathcal{G}_2(\kappa^*)$ is found and added to $\kappa^*$. This process is repeated for each layer $\mathcal{G}_i, i\leq h$ of the graph to find the selected maximal weight clique $\kappa^*\in\mathcal{G}(t)$ to be transmitted at time $t$.\ignore{ The entire algorithm structure is illustrated in Algorithm \ref{algo1}.}

Since finding the maximum weight clique in the G-IDNC graph is NP-hard \cite{arg1} we will use the same simple vertex search approach with modified weights introduced in \cite{arg2} after tailoring the weights to the ones defined in \eref{eq:weights}. Let $w_{ij}$ be the modified weights, which can be expressed as:
\begin{align}
w_{ij} = (w_{ij}^* + 1) \times \sum_{v_{kl} \in \nu(v_{ij})} w_{kl}^*\;,
\label{omegamax}
\end{align}
where $\nu(v_{ij})$ is the set of adjacent vertices of $v_{ij}$ within its layer.
\begin{algorithm}[t]
\begin{algorithmic}
\REQUIRE $\mathbf{F}$, $p_i \text{ and } C_{i}(t-1),~\forall~ i\in\mathcal{M}$.
\STATE Initialize $\kappa^* =\varnothing$.
\STATE Construct $\mathcal{G}_1\left(\mathcal{V}_1,\mathcal{E}_1\right), \mathcal{G}_2\left(\mathcal{V}_2,\mathcal{E}_2\right), ..., \mathcal{G}_h\left(\mathcal{V}_h,\mathcal{E}_h\right)$.
\FOR{l=1 \TO h}
\STATE $\mathcal{G} \leftarrow \mathcal{G}_l$.
\FORALL{ $v \in \kappa^*$}
\STATE Sets $\mathcal{G} \leftarrow \mathcal{G}(\kappa^*)$.
\ENDFOR
\WHILE{$\mathcal{G} \neq \varnothing$}
\STATE Compute $w_{ij}^*$ and $w_{ij}$ using \eref{eq:weights} and \eref{omegamax}.
\STATE Select $v^* =\underset{v_{ij}\in\mathcal{G}}{\text{argmax}} \left\{w_{ij}\right\}$.
\STATE Sets $\kappa^* \leftarrow \kappa^* \cup v^*$.
\STATE Sets $\mathcal{G} \leftarrow \mathcal{G}(\kappa^*)$.
\ENDWHILE
\ENDFOR
\end{algorithmic}
\caption{Maximum Weight Vertex Search Algorithm}
\label{algo1}
\end{algorithm}

\subsection{Binary Particle Swarm Optimization Solution}

Particle swarm Optimization (PSO) is a population based search algorithm based on the simulation of the social behavior of animals. It was introduced in \cite{488968,494215,Kennedy} by Eberhart and Kennedy for continuous functions. To each individual in the multidimensional space is associated two vectors: the position vector and the velocity vector. The velocity vector determines in which direction the position vector should evolve. 

In \cite{637339}, the authors extended their algorithm to the binary optimization and in \cite{4433821}, Khanesar proposed a novel binary PSO. This algorithm is based on a new definition for the velocity vector of binary PSO. In order to state the outlines of this algorithm, we first introduce the following quantities: Let $X_i$ be the position of particle $i$, $P_{ibest}$ is the best position particle $i$ visited and $P_{gbest}$ the best position visited by any particle. Let $L$ be the total number of particles and $T$ the number of iteration of the BPSO algorithm.

Let $V_{ij}^1$ and $V_{ij}^0$ be the probabilities that a bit $j$ of the particle $i$ changes from $0$ to $1$ and from $1$ to $0$, respectively. These probabilities are computed using the following expressions:
\begin{align}
V_{ij}^1 = wV_{ij}^1+d_{ij,1}^1+d_{ij,2}^1 \\
V_{ij}^0 = wV_{ij}^0+d_{ij,1}^0+d_{ij,2}^0,
\end{align}
where $w$ is a random number in the range $[-1,1]$ chosen at the beginning of the BPSO algorithm and $d_{ij,1}$, $d_{ij,2}$, $d_{ij,1}^0$, and $d_{ij,2}^0$ are extracted using the rule below:
\begin{align}
\text{If } P_{ibest}^j=1 \text{ Then } d_{ij,1}^1=c_1r_1 \text{ and } d_{ij,1}^0=-c_1r_1 \nonumber \\
\text{If } P_{ibest}^j=0 \text{ Then } d_{ij,1}^0=c_1r_1 \text{ and } d_{ij,1}^1=-c_1r_1 \nonumber \\
\text{If } P_{ibest}^j=1 \text{ Then } d_{ij,2}^1=c_2r_2 \text{ and } d_{ij,2}^0=-c_2r_2 \nonumber \\
\text{If } P_{ibest}^j=1 \text{ Then } d_{ij,2}^0=c_2r_2 \text{ and } d_{ij,2}^0=-c_2r_2
\end{align}
where $c_1$ and $c_2$ are fixed factor determined by user and $r_1$ and $r_2$ are two random variable in the range $[0,1]$ chosen at each iteration of the BPSO algorithm. The velocity $V_{ij}^c$ of bit $j$ of particle $i$ is defined as:
\begin{align}
V_{ij}^c = 
\begin{cases}
V_{ij}^1 \text{ if } x_{ij} = 0 \\
V_{ij}^0 \text{ if } x_{ij} = 1
\end{cases}
\end{align}
The normalized velocity is obtained using the sigmoid function. This function is defined as follows:
\begin{align}
\text{sig}(x) = \cfrac{1}{1+e^{-x}}
\end{align}
Let $V_{ij}^\prime$ be the normalized velocity (i.e. $V_{ij}^\prime= \text{sig}(V_{ij}^c)$). The position update of the bit $j$ of particle $i$ is computed as follows:
\begin{align}
x_{ij}(t+1) = 
\begin{cases}
\overline{x}_{ij}(t)\text{if }& r_{ij}< V_{ij}^\prime \\
x_{ij}(t)  \text{if }& r_{ij} > V_{ij}^\prime
\end{cases} 
\end{align}
where $\overline{x}$ is the binary complementary of $x$ and $r_{ij}$ is a random number in the range $[0,1]$ chosen at each iteration and at each bit of the particle.

Let $\mathcal{F}$ be the function of interest that we want to minimize. The outline of the BPSO algorithm are the following:
\begin{enumerate}
\item Initialize the swarm $X_i$, the position of particles are randomly initialized within the hyper-cube.
\item Evaluate the performance $\mathcal{F}$ of each particle, using its current position $X_i(t)$.
\item compare the performance of each individual to its best performance so far: if $(\mathcal{F}(X_i(t)) < \mathcal{F}(P_{ibest}))$): \\
$\mathcal{F}(P_{ibest}) = \mathcal{F}(X_i(t))$ \\
$P_{ibest}  =X_i(t)$
\item Compare the performance of each particle to the global best particle:if $(\mathcal{F}(X_i(t)) < \mathcal{F}(P_{gbest}))$): \\
$\mathcal{F}(P_{gbest}) = \mathcal{F}(X_i(t))$ \\
$P_{gbest}  =X_i(t)$
\item change the velocity of the particle, $V_i^0$ and  $V_i^1$.
\item Calculate the velocity of change of the bits, $V_i^c$.
\item Generate the random variable $r_{ij}$ in the range: (0,1). Move each particle to a new position.
\item Go to step 2, and repeat $(T-1)$ times.
\end{enumerate}
Let $\phi_i(\kappa,t)$ be the weight of user $i$ when sending the packet combination $\kappa$ at time $t$. Following the result of the previous section, $f_i(\kappa,t)$ is expressed as:
\begin{align}
&\phi_i(\kappa,t) = \\
&\begin{cases}
 \text{log}\left( \cfrac{1}{ e_i(t)}\right)&\text{if } i \in \tau \cap \mathcal{P}(t)\\
 0 & \text{otherwise } 
\end{cases} \nonumber
\end{align}

As for the multi-layer solution in the maximum weight clique problem, we want to include users that are not in the critical layer under the condition that their inclusion to not disturb the instantaneous decodability of the targeted users in the critical layer. This inclusion is done using layer prioritization i.e. the inclusion of a user in layer $\mathcal{P}_m$ should not bother the instant dependability of users in layer $\mathcal{P}_n, n>m$. To reproduce this concept of prioritization using a single real function, we use the sigmoid function. The new objective function to maximize is the following:
\begin{align}
\phi_i^\prime(\kappa,t) = \sum_{i \in M_w )} \text{sig}(\tilde{\phi}_i(\kappa,t))+M*(h-P(i)) 
\end{align}
where $h$ is the total number of layers, $P(i)$ is the index of the layer of user $i$ and $\tilde{\phi}_i$ defined as follows: 
\begin{align}
\tilde{\phi}_i(\kappa,t) =  \text{log}\left( \cfrac{1}{ e_i(t)}\right)
\end{align}

The sigmoid function is an increasing function between $[0 ,1]$ and therefore the utility of each user is shifted according to the number of the layer he is in. Since no more that $M$ players can be simultaneously on the same layer, then all these layers are not overlapping and a single element of one layer will yield a better objective that all the sum of all the objective of user in layer below him. This function respects the prioritization as the multi-layer graph. The function of interest to minimize , at each time instant $t$, can be written as $\mathcal{F}(\kappa) = \phi_i^\prime(\kappa,t)$.
The following theorem gives the convergence of the whole system in our cases:
\begin{theorem}
For a number of particle equal to the number of packets (i.e. $L=N$) and a sparse initialisation vector, the overall system will converge  
\end{theorem}
\begin{proof}
The proof of this Theorem can be found in Appendix C.
\end{proof}

\section{Extension To The limited Feedback Scenario}\label{sec:ext}

In this section, we extend our previous analysis to the limited feedback scenario. First, we introduce the system model, the backward channel model and the feedback protocol. We, then, derive the expression of the optimal packet combination to reduce the completion time in such lossy feedback environment. Finally, we propose modified version of previously introduced algorithm to effectively reduce the completion time.

\subsection{System Model}

Due to the feedback loss that can occur in the system, at the end of the \emph{initial phase}, three sets of packets are attributed to each user $i$ at the sender:
\begin{itemize}
\item The \emph{Has} set (denoted by $\mathcal{H}_i(t)$) is defined as the set of packets successfully received by user $i$.
\item The \emph{Wants} set (denoted by $\mathcal{W}_i(t)$) is defined as the set of packets that are lost by user $i$. In other words,we have $\mathcal{W}_i = \mathcal{N} \setminus \mathcal{H}_i$.
\item The Uncertain set (denoted by $\mathcal{U}_i(t)$). It is the sets of packets which the BS is not certain if either the packet or the feedback were erased. We have $\mathcal{U}_i(t) \subseteq \mathcal{W}_i(t)$.
\end{itemize}

The BS saves the information obtained after the transmission at time $(t-1)$ in a \emph{feedback matrix} (FM) $\mathbf{F}(t) = [f_{ij}(t)],~ \forall~ i \in \mathcal{M},~ \forall~j \in \mathcal{N},~ \forall~t > 0$ such that:
\begin{align}
f_{ij}(t) =
\begin{cases}
0 \hspace{0.9 cm}& \text{if } j \in \mathcal{H}_i(t) \\
1 \hspace{0.9 cm}& \text{if } j \in \mathcal{W}_i(t) \setminus \mathcal{U}_i(t) \\
x \hspace{0.9 cm}& \text{if } j \in \mathcal{W}_i(t) \cap \mathcal{U}_i(t).
\end{cases}
\end{align}

\subsection{Backward Channel Model and Feedback Protocol}

We model the backward channel (from the users to the BS), as for the forward channel, by a Gilbert-Elliott channel. The parameters of the backward channel are defined in the same way as those of the forward channel with a superscript $b$. For example, $C_i^b$ will denote the state of the forward channels. For each user $i$, the memory factor of the backward channel $\psi_i$ and the average memory can be formulated as:
\begin{align}
\psi_i = 1 - g_i^b - b_i^b \nonumber\\
\psi = \cfrac{\sum\limits_{i \in \mathcal{M}} \psi }{ M}.
\end{align}

When both the forward and the backward channel have the same transitions probabilities (i.e. $g_i^f = g_i^b$ and $b_i^f = b_i^b, \forall~ i \in \mathcal{M}$), the channel is said to be identically distributed and when they are experiencing the same channel realization (i.e $C_i^f(t) = C_i^b(t),\forall~ t >0$) the channel is said to be reciprocal.

Before transmitting the data packet, the BS first transmits $N$ bits representing the source packets that are combined in the encoded packet (i.e. if the packet is in the combination the corresponding bit is set to $1$ otherwise it is set to $0$). Each user, after listening to the header, continues to listen if the packet combination is instantly decodable for him. Otherwise, energy is saved by ignoring the next transmission. The targeted users by a transmission are the users that can instantly decode a packet from that transmission. 

Time is divided into frames of same length equal to $T_f$ time slots. Each of these frame is composed of a downlink sub-frame and an uplink sub-frame. The length of the downlink and uplink sub-frame are respectively $T_d$ and $T_u$. In other words, we have $T_f = T_d+T_u$. The BS transmits the packet combinations during the downlink sub-frame and do not get any feedback the meanwhile. During the uplink sub-frame, the BS listens to the feedbacks sent from the different users and do not transmit packets. After each downlink sub-frame, users that received and managed to decode the packet combination during that downlink sub-frame acknowledges its reception by sending a feedback during the uplink sub-frame. Define $T_{u_i}$ as the time-slot in the uplink sub-frame in which user $i$ is able to send feedback (i.e. $1 \leq T_{u_i} \leq T_u ,~ \forall~ i \in \mathcal{M}$). In other words, user $i$ can transmit acknowledgement, during frame number $n$, at time $t=nT_f-T_u+T_{u_i}$. Only targeted users during the downlink sub-frame will send acknowledgement. In other words, if a feedback from one of the targeted user is lost, the BS will not get any feedback from this user until the next transmission in which it is targeted. Each feedback sent from a user consists of $N$ bits indicating all previously received/lost packets. As a consequence, after receiving feedback from an arbitrary user, the states of all its packet in the FM will be certain (i.e. $\mathcal{U}_i = \varnothing$). The BS uses the feedback to update the feedback matrix. This process is repeated until all users report that they obtained all the wanted packets. 

In order to be able to accurately estimate the state of the forward/backward channels for each user, we impose, as in \cite{refjournal}, that for each targeted user, from the last time a feedback is heard from that user, there will be at least one packet which is attempted only once. This constraint becomes unnecessary for user that still need only a single source packet. An additional $log_2(T_d)$ are send with the feedback indication the decoding delay encountered by the user during the downlink sub-frame.

\subsection{Transmission/Feedback Loss Probabilities at Time $t$}

In this section, we compute the probability to loose the transmission $e_i(t) \triangleq\mathds{P}(C_i^p(t) = B) ,~ \forall~ i \in \mathcal{M}$ and to loose the feedback $f_i(t) \triangleq \mathds{P}(C_i^q(t) = B) ,~ \forall~ i \in \mathcal{M}$, at time slot $t$.

In order to compute these probabilities, we first introduce the following variables: Let $n_i^{(-1)}$ and $n_i^{(0)}$ ($n_i^{(-1)} < n_i^{(0)}$) be the indices of the most recent frame where the sender heard a feedback from user $i$. Let $\lambda_{ij}(n)$ be the set of the time indices when packet $j$ was attempted to user $i$ during frame number $n$. Define $j_i^{(0)}$ as the last sent packet among those which were attempted only once between the two frames $(n_i^{(-1)}+1)$ and $n_i^{(0)}$ to user $i$. This variable can be mathematically defined as:
\begin{align}
j_i^{(0)} &= \underset{j \in \mathcal{W}_i(n_i^{(0)} \times T_f)}{\text{argmax}}  \bigcup_{k=n_i^{(-1)}+1}^{n_i^{(0)}} \lambda_{ij}(k) \nonumber \\
& \quad \text{subject to } \left| \bigcup_{k=n_i^{(-1)}+1}^{n_i^{(0)}} \lambda_{ij}(k) \right| = 1 ,
\end{align}
where $\cup_{x \in X}A_x $ is the union of the sets $A_x,~ \forall~ x \in X$. Let $t_i^{(0)}$ be the time where packet $j_i^{(0)}$ was attempted to receiver $i$ and $t_i^*$ the last time the a feedback was heard from user $i$. In other words:
\begin{align}
t_i^{(0)} &= \bigcup_{k=n_i^{(-1)}+1}^{n_i^{(0)}} \lambda_{ij_i^{(0)}}(k) \\
t_i^* &= n_i^{(0)} \times T_f - T_u + T_{u_i}.
\end{align}

Given these definitions, we can introduce the following theorem regarding  the loss probabilities of the forward $e_i(t)$ and feedback $f_i(t)$ transmissions at any given time $t$.
\begin{theorem}
The probabilities $e_i(t)$ of loosing a transmission from receiver $i$ at time $t>t_i^*$ can be expressed as:
\begin{align}
&e_i = \\
&\begin{cases}
&\cfrac{p_i^fg_i^f}{p_i^fg_i^f+q_i^fb_i^f}(p_i^f+(q_i^f-p_i^f)b_i^f\sum\limits_{i=0}^{t-t_i^{(0)}-1} \mu^{i})  \nonumber \\ 
& + \cfrac{q_i^fb_i^f}{p_i^fg_i^f+q_i^fb_i^f}(q_i^f+(p_i^f-q_i^f)g\sum\limits_{i=0}^{t-n^0-1} \mu^{i}) \nonumber \\ 
& \qquad \text{ if } f_{ij^0} = 1 \nonumber \\
&\cfrac{(1-p_i^f)g_i^f}{(1-p_i^f)g_i^f+(1-q_i^f)b_i^f} \nonumber \\
& \hspace{2cm} \times  (p_i^f+(q_i^f-p_i^f)b_i^f\sum\limits_{i=0}^{t-t_i^{(0)}-1} \mu^{i})  \nonumber \\ 
& + \cfrac{(1-q_i^f)b_i^f}{(1-p_i^f)g_i^f+(1-q_i^f)b_i^f} \nonumber \\
& \hspace{2cm} \times  (q_i^f+(p_i^f-q_i^f)g_i^f\sum\limits_{i=0}^{t-t_i^{(0)}-1} \mu^{i}) \nonumber \\ 
& \qquad  \text{ if } f_{ij^0} = 0
\end{cases}
\end{align}
The probabilities $f_i(t)$ of loosing a feedback from user $i$ at time $t>t_i^*$ can be expressed as:
\begin{align}
&f_i = \cfrac{(1-p_i^b)g_i^b}{(1-p_i^b)g_i^b+(1-q_i^b)b_i^b} (p_i^b+(q_i^b-p_i^b)b_i^b\sum_{i=0}^{t-t_i^*-1} \psi^{i})  \nonumber \\ 
& + \cfrac{(1-q_i^b)b_i^b}{(1-p_i^b)g_i^b+(1-q_i^b)b_i^b}(q_i^b+(p_i^b-q_i^b)g_i^b\sum_{i=0}^{t-t_i^*-1} \psi^{i})   
\end{align}
\label{sffrh}
\end{theorem}
\begin{proof}
The proof can be found in Appendix D.
\end{proof}

\subsection{Problem Formulation}

In order to state the optimization problem we first define this two probability : 
\begin{itemize}
\item The innovative probability $p_{i,n}(j,t)$: probability that packet $j$ is innovative for user $i$ at time $t$.
\item The finish probability $p_{i,f}(t)$: probability that user $i$ successfully received all his primary packets but $\mathcal{W}_i(t) \neq \varnothing$ at time $t$.
\end{itemize} 
The expressions of these probabilities is available in the Appendix E.

The following theorem defines a maximum weight clique algorithm that can satisfy the critical criterion.
\begin{theorem}
The critical criterion in \eref{eq:heuristic-criterion} can be achieved by selecting $\kappa^*(t)$ according to the following optimization problem:
\begin{align}
&\kappa^{*}(t) \nonumber \\
&= \underset{\kappa(t) \in \mathcal{G}}{\text{argmax}} \left\{ \sum_{i \in (\mathcal{P}(t) \cap \tau ) } \text{log}\left( 1 + \cfrac{p_{i,n}(\kappa_i,t)}{ \cfrac{e_i(t)}{1-e_i(t)}+p_{i,f}(t)}\right) \right\}.
\end{align}
In other words, the transmission $\kappa(t)$ than can satisfy the critical criterion can be selected using a maximum weight clique problem in which the weight of each vertex $v_{ij}$ in $\mathcal{P}(t)$ can be expressed as:
\begin{align}\label{degwe}
w_{ij}^* = \text{log}\left( 1 + \cfrac{p_{i,n}(j,t)}{ \cfrac{e_i(t)}{1-e_i(t)}+p_{i,f}(t)}\right).
\end{align}
\end{theorem}
\begin{proof}
The proof can be found in Appendix F.
\end{proof}

\subsection{Proposed Algorithm}

\subsubsection{Maximum Weight Clique Solution}

In order to minimize the completion time in G-IDNC, we look for all possible combinations of source packets then take the combination of these packets that guarantees the minimum delay for this transmission. To represent all the feasible packet combinations, we use the graph model introduced in \cite{letterarxiv} and called the lossy D-IDNC graph (LG-IDNC). To construct the LG-IDNC graph, we first introduce the expected decoding delay increase $d_{ij,kl}(j \oplus l)$ for two distinct arbitrary users $i$ and $k$ after sending the packet combination $j \oplus l$:
\begin{align}
\label{joplusl}
&d_{ij,kl}(j \oplus l) =  \\
&  (1-e_i)(p_{n,i}(j)p_{n,i}(l)+ (1-p_{n,i}(j)(1-p_{n,i}(l)))\overline{p}_{f,i}+  \nonumber \\
& (1-e_k)(p_{n,k}(j)p_{n,k}(l)+ (1-p_{n,k}(j)(1-p_{n,k}(l)))\overline{p}_{f,k},\nonumber
\end{align}
where $\overline{p}_{f,i} = 1 - {p}_{f,i}$.

To obtain the expected decoding delay increase $d_{ij,kl}(j)$ for these users after sending packet $j$, we replace $l$ by $0$ in \eref{joplusl} and take $p_{n,i}(0) =0,~\forall~i \in \mathcal{M}$. To construct the LG-IDNC graph $\mathcal{G} (\mathcal{V},\mathcal{E})$, we first create a vertex $v_{ij}  ,~\forall~ i \in \mathcal{M},~\forall~ j \in \mathcal{W}_i$. We then connect two users, if they either need the same packet or if the expected decoding delay increase is lower when sending the packet combination than when sending only one packet. In other words, we connect by an edge two vertices $v_{ij}$ and $v_{kl}$ if one of the following conditions is true:
\begin{itemize}
\item C1: $j=l \Rightarrow$ Packet $j$ is needed by both the users $i$ and $k$.
\item C2: $d_{ij,kl}(j \oplus l) \leq min (d_{ij,kl}(j) ,d_{ij,kl}(l) ) \Rightarrow$ The packet combination $j \oplus l$ guarantees a lower decoding delay to the users $i$ and $k$ than packets $j$ and $l$ individually.
\end{itemize}

Unlike condition C1 that does not require packet combination, C2 involves the combination of packet $j$ and $l$. It was shown in \cite{arg1} that in the perfect feedback scenario, all possible packet combinations is equivalent to a maximal weight clique in the LG-IDNC graph. Therefore the combination that reduces the best the completion time for the current transmission is the maximum weight clique in the LG-IDNC graph. This result can be extended to the lossy intermittent feedback. The BS generate the encoded packet by taking the binary XOR of packet represented by the maximal weight clique in the LG-IDNC graph. Users targeted by this transmission are those represented by this maximal weight clique.

After definition of the LG-IDNC graph, we use the multi-layer algorithm developed in the previous \sref{sec:alg} with the new weights developed in \eref{degwe}.

\subsubsection{BPSO Algorithm}

The same line of thinking used in \sref{sec:alg} applies in the case of the limited feedback scenario. We objective function is defined using the new weights \eref{degwe} to reflect the uncertainties in the system. The new objective function to maximize is the following:
\begin{align}
\phi_i^\prime(\kappa,t) = \sum_{i \in M_w )} \text{sig}(\tilde{\phi}_i(\kappa,t))+M*(h-P(i)) 
\end{align}
where $h$ is the total number of layers, $P(i)$ is the index of the layer of user $i$ and $\tilde{\phi}_i$ defined as follows: 
\begin{align}
\tilde{\phi}_i(\kappa,t) =  \text{log}\left( 1 + \cfrac{p_{i,n}(\kappa,t)}{ \cfrac{e_i(t)}{1-e_i(t)}+p_{i,f}(t)}\right)
\end{align}

\subsection{Blind Graph Policies Solution}

In the lossy intermittent feedback, uncertainties about the reception state of the different packets and the decodability conditions make the algorithm proposed in \cite{vtc} non effective to actually reduce the completion time. To solve this problem, we introduce three partially blind algorithms that estimate all the uncertain packets with a predefined policy, update the graph accordingly and finally perform packet selection using the algorithm proposed in \cite{vtc}. These graph update approaches are the generalization of the update methods proposed in \cite{ref13} in the context of reducing the completion time with lossy feedback.

\subsubsection{Pessimist Graph Update}

In this approach, all packets that are not fed back by users are considered erased rather than assuming that their feedback is erased. Reconsidering these packets in the following transmissions gives them a greater chance to be reattempted rapidly. Since no acknowledgement is expected to be heard during the downlink sub frame, packets attempted meanwhile are systematically not reconsidered in the following transmissions. If a feedback is heard in the uplink sub frame, the state of the user is updated. Otherwise, all the uncertain packet of that user are reconsidered.

In the pessimist graph update approach, uncertain vertices are removed from the graph during the downlink sub frame and reconsidered in the uplink sub frame if no acknowledgement is heard from the concerned user.

\subsubsection{Optimist Graph Update}

In this approach, all packets that are not fed back by users are considered received and their corresponding feedback erased. Not reconsidering these packets in the following transmissions gives a greater chance to non attempted packets to be transmitted. Since no feedback can be heard from a user having all its packets in an uncertain state, unless this user is targeted. therefore, users with full uncertain Wants set are reconsidered after the uplink sub frame.

In the optimist graph update approach, uncertain vertices are removed from the graph and reconsidered in the uplink sub frame if the user have full uncertain Wants set.

\subsubsection{Realistic Graph Update}

In this approach, all packets that are not fed back by users are probabilistically considered received and their acknowledgement erased and reciprocally. This approach tends to stochastically balance between reattempting packets with unheard feedback and transmitting new packets. Since during the downlink sub frame no acknowledgement is expected to be heard, then packets are reconsidered with probability $\mathcal{P}_{B_i^f}$ and discarded with probability $\mathcal{P}_{G_i^f}$. In the uplink frame, all the uncertain packets are reconsidered with probability $\mathcal{P}_{B_i^b}$ and removed with probability $\mathcal{P}_{G_i^b}$.

In the realist graph update approach, uncertain vertices are removed from the graph with probability $\mathcal{P}_{G_i^f}$ in the downlink frame and with probability $\mathcal{P}_{G_i^b}$ in the uplink frame.

\section{Simulation Results}\label{sec:sim}

\begin{figure}[t]
\centering
  \includegraphics[width=1\linewidth]{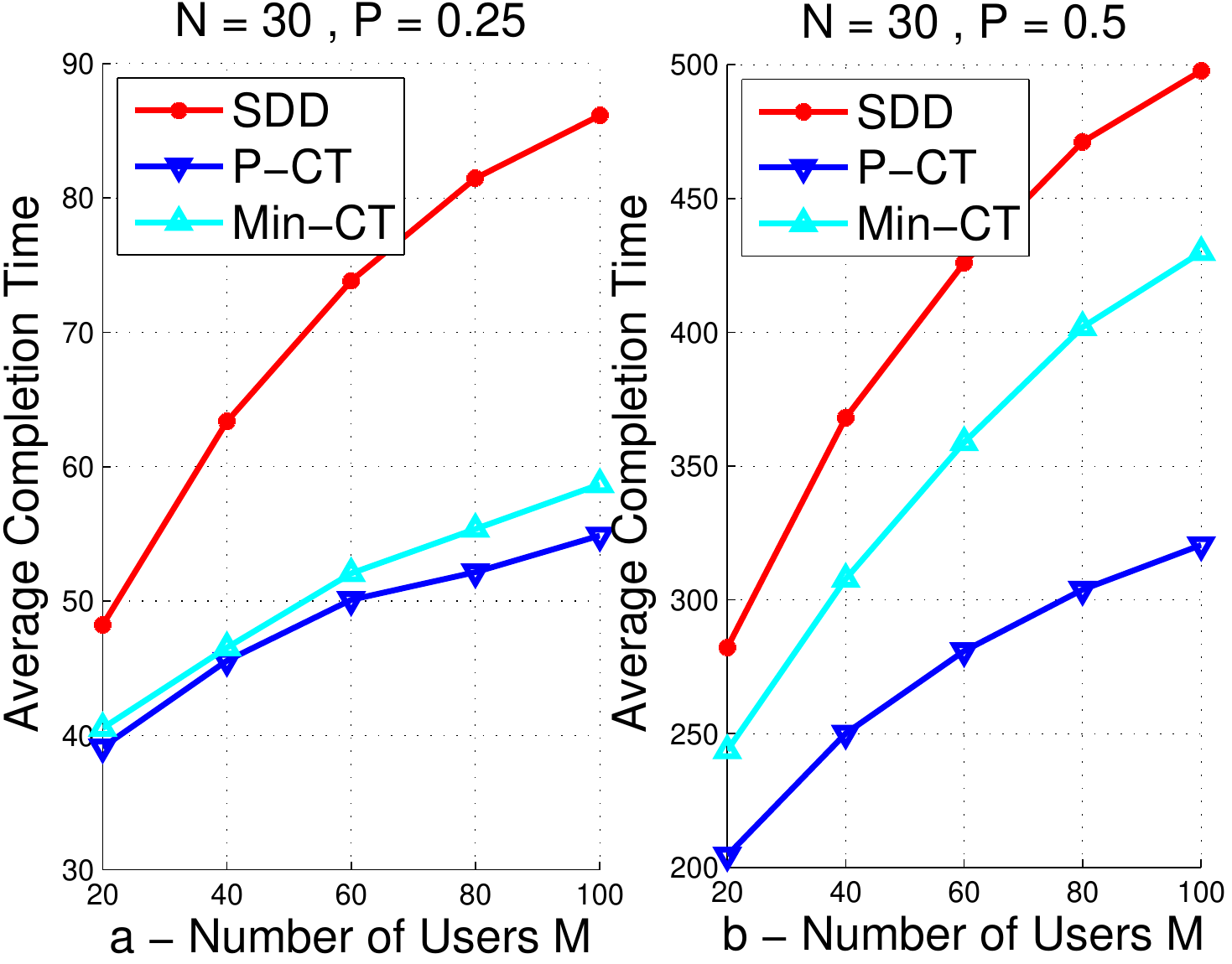}\\
  \caption{Mean completion time for G-IDNC versus number of users $M$.}\label{fig:M}
\end{figure}

\begin{figure}[t]
\centering
  \includegraphics[width=1\linewidth]{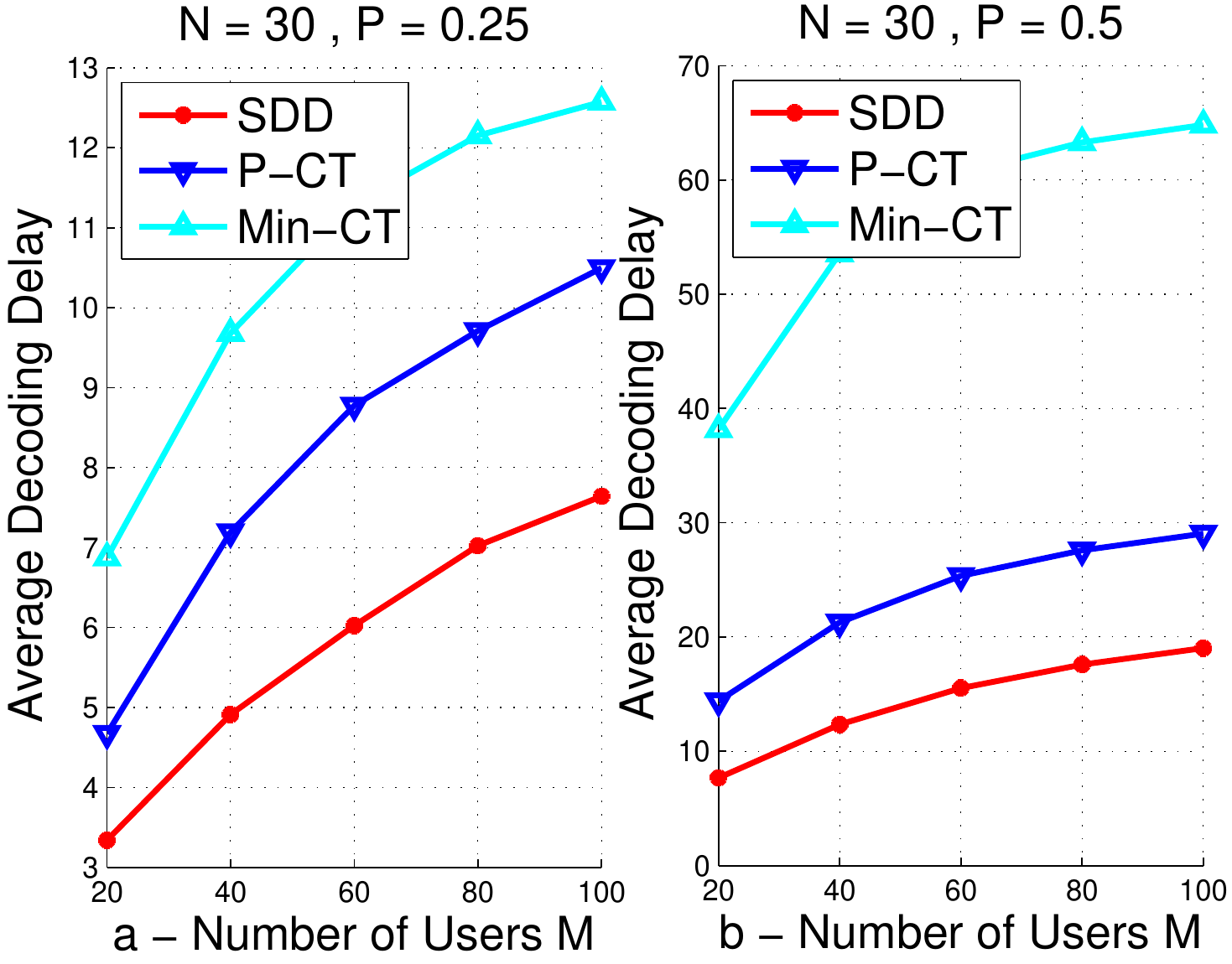}\\
  \caption{Mean decoding delay for G-IDNC versus number of users $M$.}\label{fig:M2}
\end{figure}

\begin{figure}[t]
\centering
  \includegraphics[width=1\linewidth]{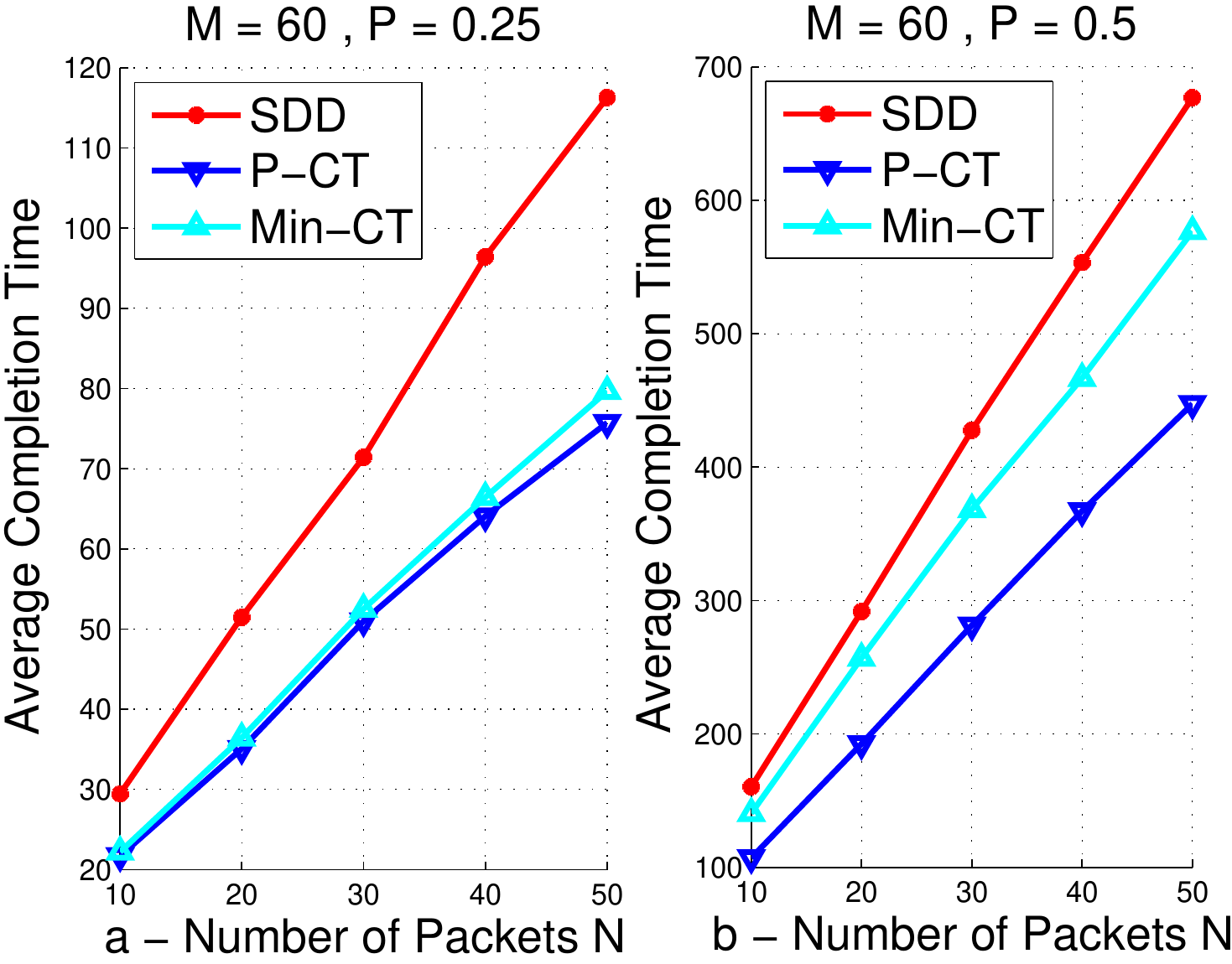}\\
  \caption{Mean completion time for G-IDNC versus number of packets $N$.}\label{fig:N}
\end{figure}

\begin{figure}[t]
\centering
  \includegraphics[width=1\linewidth]{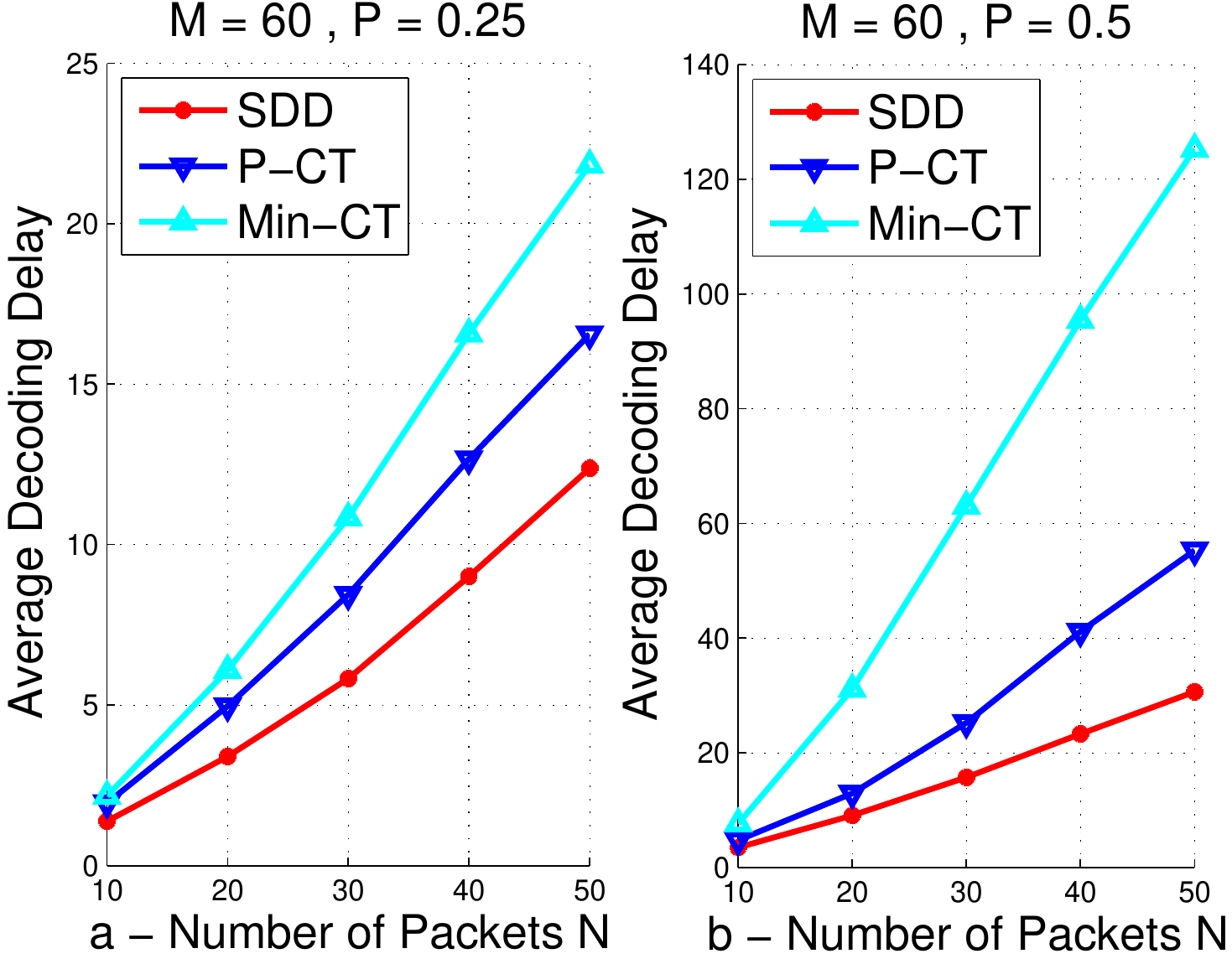}\\
  \caption{Mean decoding delay for G-IDNC versus number of packets $N$.}\label{fig:N2}
\end{figure}

\begin{figure}[t]
\centering
  \includegraphics[width=1\linewidth]{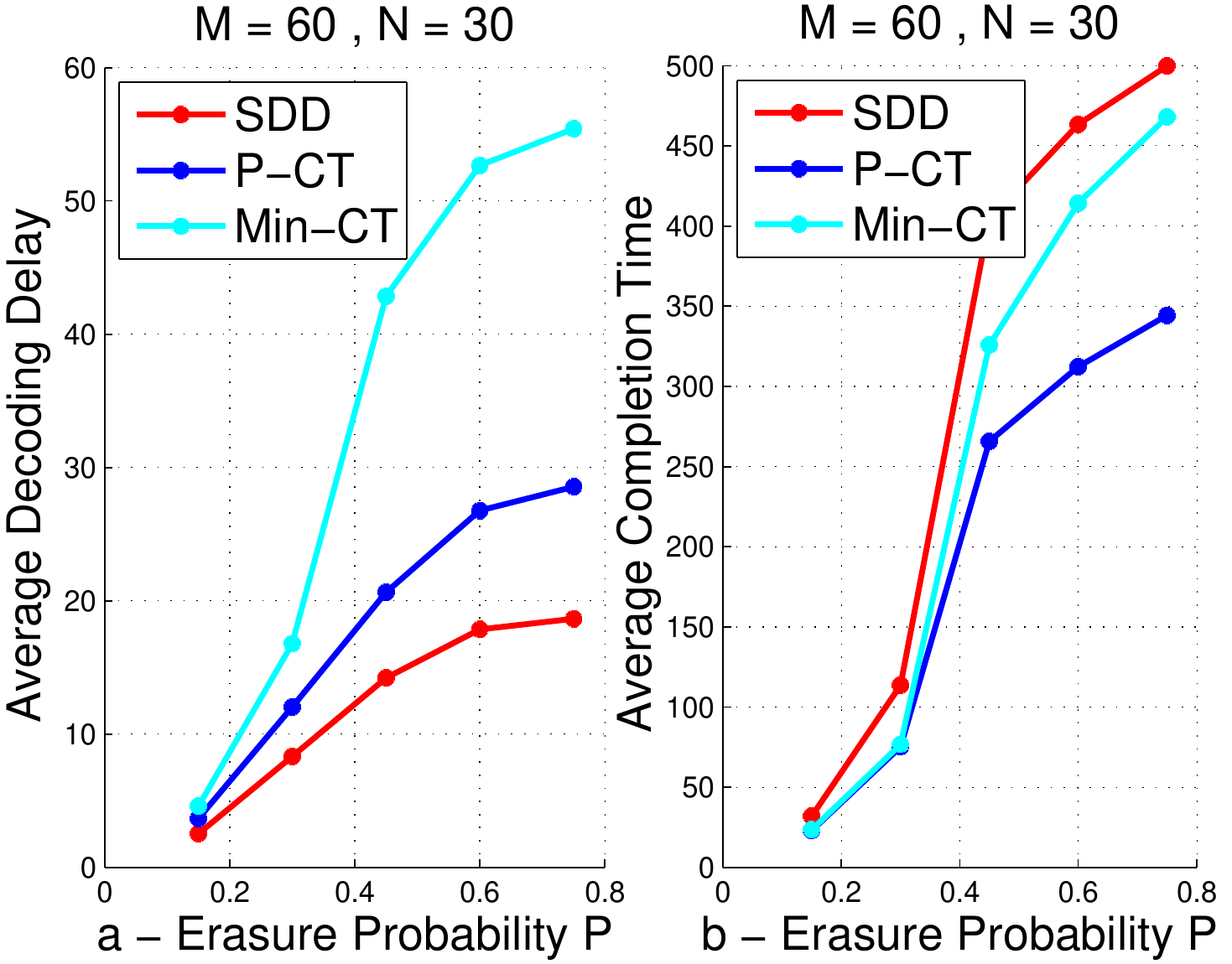}\\
  \caption{Mean delays for G-IDNC versus packet erasure probability $P$.}\label{fig:P}
\end{figure}

\begin{figure}[t]
\centering
  \includegraphics[width=1\linewidth]{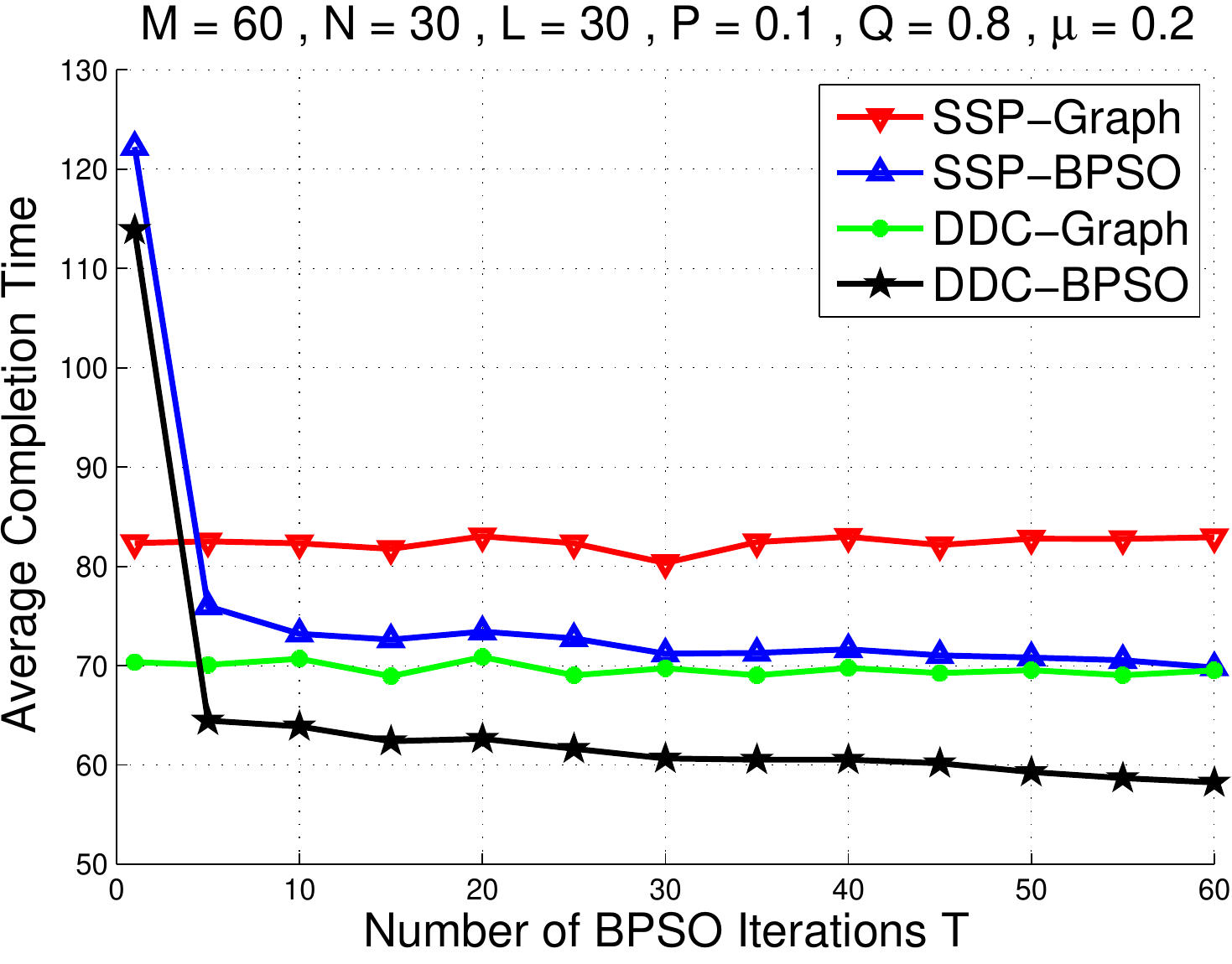}\\
  \caption{Average completion time versus the number of iteration of BPSO $T$.}\label{fig:ITSWAP}
\end{figure}
\begin{figure}[t]
\centering
  \includegraphics[width=1\linewidth]{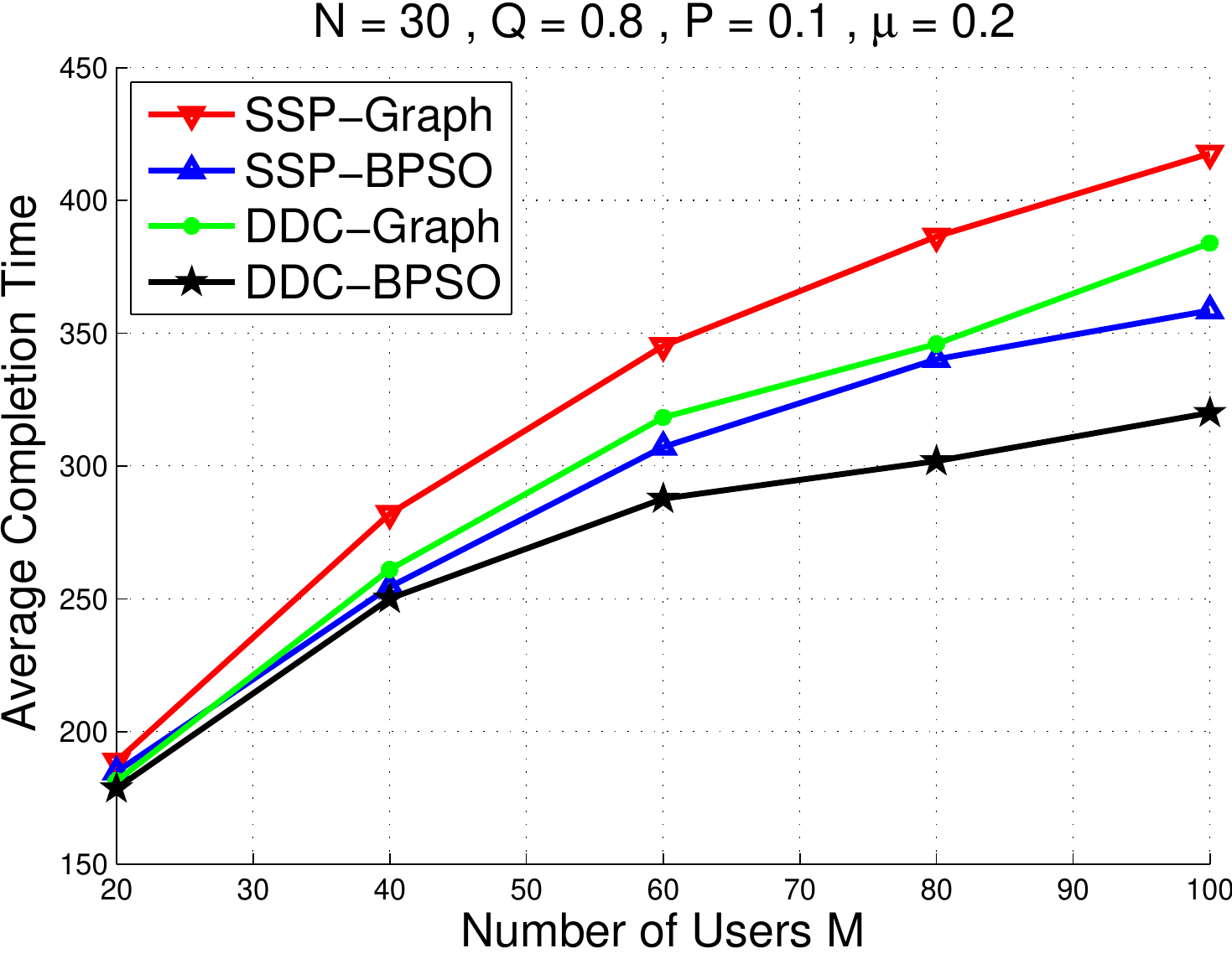}\\
  \caption{Average completion time versus number of users $M$.}\label{fig:MGSWAP}
\end{figure}
\begin{figure}[t]
\centering
  \includegraphics[width=1\linewidth]{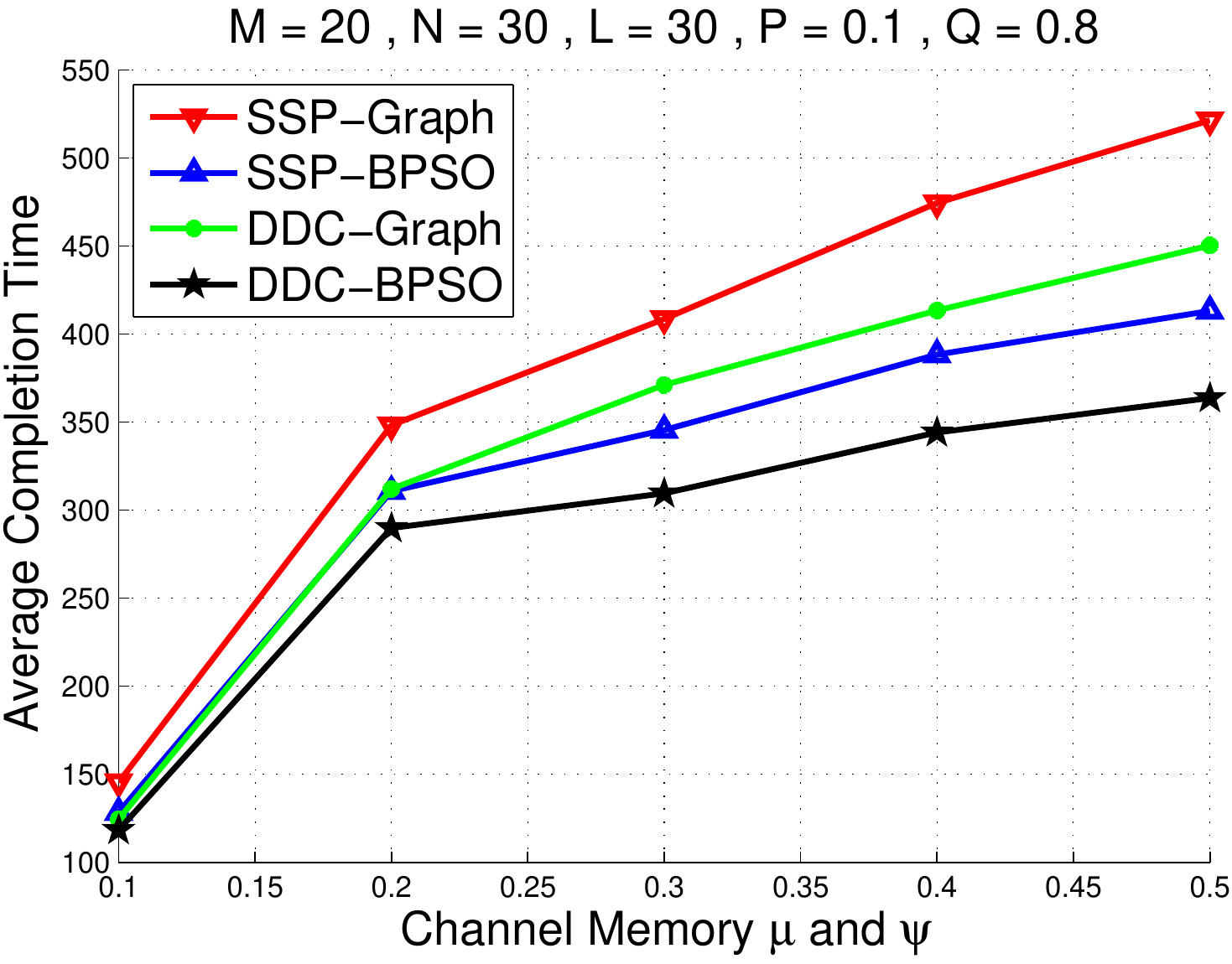}\\
  \caption{Average completion time versus chanel memory $\mu$ and $\psi$.}\label{fig:MUSWAP}
\end{figure}
\begin{figure}[t]
\centering
  \includegraphics[width=1\linewidth]{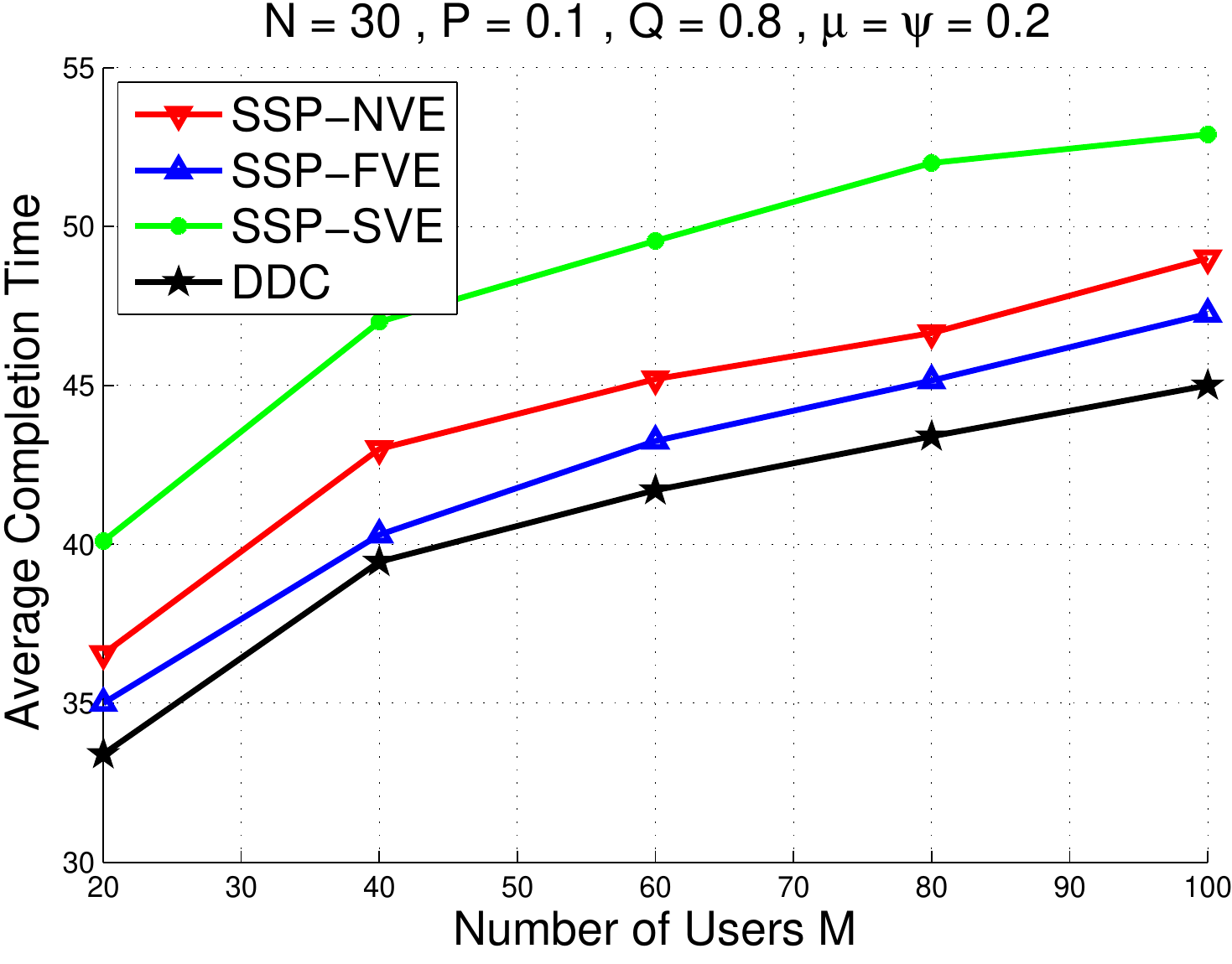}\\
  \caption{Average completion time versus the number of users $M$.}\label{fig:MG}
\end{figure}
\begin{figure}[t]
\centering
  \includegraphics[width=1\linewidth]{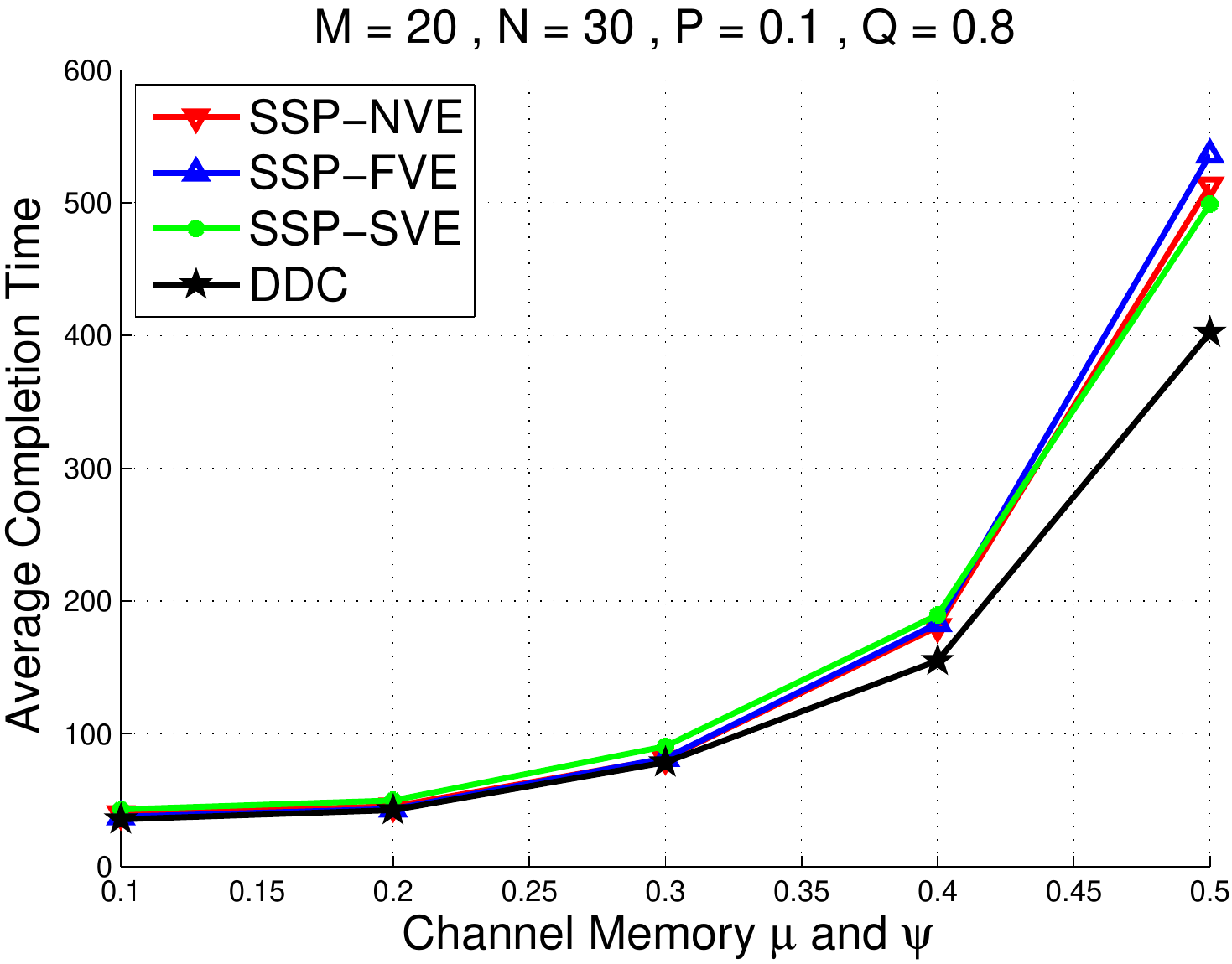}\\
  \caption{Average completion time versus chanel memory $\mu$ and $\psi$.}\label{fig:MUG}
\end{figure}

In this section, we first present the simulation results comparing the completion time and the decoding delay achieved by the different policies to optimize each in perfect feedback and independent erasure channels. We, then, present the completion time achieved by our policy against the blind policies for a lossy feedback and persistent erasure channels. In the first part, we compare, through extensive simulations the sum decoding delay (denoted by SDD) and the completion time achieved by \cite{ref4} (denoted by Min-CT) and the completion time achieved by our algorithm (denoted by P-CT) for perfect feedback while using the policy to reduce the sum decoding delay \cite{ref2} and the policy \cite{ref4} and our policy to reduce the completion time \cite{letterarxiv}. In the second part, we first compare our two heuristics to reduce the completion time for perfect feedback and persistent erasure channels. We, then, compare the completion time achieved by our policy and the blind policies in lossy feedback environment.

In all the simulations, the different delays are computed by frame then averaged over a large number of iterations. We assume that the packet and the feedback erasure probability of all the users change from frame to frame while the average packet erasure probability remain constant. We further assume the symmetric channels for both the forward and backward links. In other words, the erasure probability on the forward and the backward link are the same (the probabilities only and not the channel realizations). 

\fref{fig:M} depicts the comparison of the mean completion achieved by the policy to reduce the sum decoding delay (SDD), \cite{ref4} policy and our one to reduce the completion time (Min-CT and P-CT) for a perfect feedback and independent erasure channels against $M$ for $N=60$ and $P=0.25$ and $P=0.5$ receptively, where $P$ refers to the average packet erasure probability in the independent erasure channels. \fref{fig:M2} illustrates the comparison of the decoding delay for the same inputs. \fref{fig:N} and \fref{fig:N2} depicts the comparison of the aforementioned delay aspects against $N$ for $M=60$ and $P=0.25$ and $P=0.5$ receptively and \fref{fig:P} illustrates this comparison against the erasure probability $P$ for $M=60$ and $N=30$.

\fref{fig:ITSWAP} illustrate the comparison of the completion time achieved by the one layered algorithm (denoted by -Graph) and the BPSO algorithm (denoted by -BPSO) for both our policy to reduce the completion time \cite{confarxiv} using the decoding delay control (denoted by DDC) and the \cite{ref4} policy using the shortest stochastic path (denoted by SSP) for a perfect feedback and persistent erasure channels against the number of iteration $T$ for $M=60$, $N=30$, $L=30$, $P=0.1$, $Q=0.8$ and $\mu=0.2$. \fref{fig:MGSWAP} and \fref{fig:MUSWAP} depicts the same comparison against the number of users and the memory of the channel respectively for $\mu=0.2$($M=60$), $N=30$, $L=30$, $P=0.1$ and $Q=0.8$.

\fref{fig:MG} shows the comparison between the different blind algorithm (denoted by NVE for the pessimist, FVE of the optimist and SVE for the stochastic policy) when using the original formulation of the weights proposed in \cite{ref4} and our multi-layer graph algorithm using the decoding delay control (denoted by DDC) against the number of users $M$ for $N=30$, $P=0.1$, $Q=0.8$, and $\mu=0.2$. \fref{fig:MUG} depicts the same comparison against the channel memory $\mu$ and $\psi$ for $M=60$, $N=30$, $P=0.1$, and $Q=0.8$.

From all the figures, we can clearly see that our proposed completion time algorithm outperforms the completion time policy proposed in \cite{ref4} and the blind policies. Moreover it gives the best agreement among the sum decoding delay and the completion in G-IDNC. The completion time policy offers, in average, the minimum sum of all the delay aspects in all situations.

\fref{fig:M}.a and \fref{fig:N}.a depicts the completion time when applying the sum decoding delay policy, the completion time policy and the our completion time policy against $M$ and $N$ for a low packet erasure probability. We see that the performance of P-CT and Min-CT are very close. Whereas in \fref{fig:M2} and \fref{fig:N2} where the sum decoding delay is computed for the same inputs, the performance of P-CT is much better than Min-CT one.

As the channel conditions become harsher (high packet erasure probability), our policy to reduce the completion time minimize the completion time better than the Min-CT. We can see from \fref{fig:M}.b, \fref{fig:M2}.b, \fref{fig:N}.ab and \fref{fig:N2}.b that P-CT outperforms Min-CT in minimizing both the sum decoding delay and the completion time. \fref{fig:P}.a shows that for $P>0.3$, P-CT achieves a significant improvement in the completion time. This can be explained by the light of the P-CT policy characteristics. In the P-CT policy, the number of the erased packets is estimated using the law of large numbers. This approximation can be effective when the erasure of the channel or the input (number of packets and users) are high enough.

From \fref{fig:ITSWAP}, we clearly see that the BPSO algorithm achieves a lower completion time for a low number of iteration ($5$ iterations). This algorithm have a fixed complexity unlike the multi-layer graph algorithm which have a worst complexity of $MN$ ($60$ in the figure). This fixed complexity property along with its performance to effectively reduce the completion time make this algorithm more reliable and more suitable to be used.

\fref{fig:MG} and \fref{fig:MUG} shows that our algorithm to control the completion time using the decoding delay in lossy feedback scenario outperforms largely the blind algorithms specially as the channel is more and more persistent ($\mu$ increases). The optimist approach (FVE) achieves a reasonable degradation for a low channel persistence. However, this degradation become more severe as the memory of the channel increases. The pessimist approach (NVE) can be seen as the complementary of the optimist approach since it perform better in high memory channel and less in near independent channel. The stochastic approach (SVE) achieves an intermediate result and degrades as the channel is near independent or highly correlated.

\section{Conclusion} \label{sec:conclusion}

In this paper, we studied the effect of controlling the decoding delay to reduce the completion time below its currently best known solution for persistent channel. We first derived the decoding-delay-dependent completion time expressions. We then employed these expressions to design two new heuristic. The first decides on coded packets by reducing the probability of decoding delay increase on a new layering of the IDNC graph based on user criticality in increasing the overall completion time and the second uses binary optimization with multi-layer objective function that preserves prioritization. We, then, extended our study to the limited feedback environment. Simulation results showed that this new algorithm achieves a lower mean completion time and mean decoding delay compared to the best known completion time heuristics, with significant gains in harsh erasure scenarios.

\appendices
\numberwithin{equation}{section}

\section{Auxiliary Theorems}

In this appendix we provide auxiliary theorem that we will use to proof Theorem \ref{sffrh}. The following theorem provides the expression of the probability to be in a state of the channel given that the channel was in a particular state at a previous time instant.
\begin{theorem}
Let $(X_n)_{n\geq1}$ be a two state ($x$ and $y$) Markov chain, with $P_{tr_{ x\rightarrow y}}$ and $P_{tr_{y\rightarrow x}}$ the transition probability from state $x$ to $y$ and $y$ to $x$, respectively. Let $\mu = (1-P_{tr_{ x\rightarrow y}}-P_{tr_{y\rightarrow x}})$ be the memory of the chain. 

Define $f(n) = \mathds{P}(X_n = y | X_{n^0} = x),~ \forall~ n \geq n^0$. We have:
\begin{align}
f(n) = P_{tr_{ x\rightarrow y}} \times \langle \sum_{i=0}^{n-n^0-1} \mu^{i} \rangle 
\end{align}
where 
\begin{align}
\langle \sum_{x \in X} (.) \rangle = 
\begin{cases}
\sum_{x \in X} (.) &\text{if } X \neq \varnothing \\
0  &\text{if } X = \varnothing.
\end{cases} 
\end{align}
\label{regj}
\end{theorem}
\begin{proof}
The proof can be found in \cite{refjournal} Appendix A.
\end{proof}
The following theorem gives the expression of the probability to be in a particular state of the channel given the channel realization for the same time instant. Lets consider the channel defined in \fref{fig:GEC}. For notation simplicity, we will not consider the superscripts in this theorem. Let $X_i(n)$ be a random variable that take the value $1$ if the transmission at time $n$ is erased and $0$ otherwise.
\begin{theorem}
The probability of the state of the channel at time $n$ conditioned by the realization $X_i(n)$ at the same time can be expressed as:
\begin{align}
&\mathds{P}(\mathcal{C}(n)=y|X_i(n)=x) = \\
&\begin{cases}
\cfrac{p\mathcal{P}_G}{p\mathcal{P}_G+q\mathcal{P}_B} &\text{ if } x = 1, y = G\\
\cfrac{(1-p)\mathcal{P}_G}{(1-p)\mathcal{P}_G+(1-q)\mathcal{P}_B}  &\text{ if } x = 0, y = G\\
\cfrac{q\mathcal{P}_B}{p\mathcal{P}_G+q\mathcal{P}_B}  &\text{ if }  x = 1, y = B\\
\cfrac{(1-q)\mathcal{P}_B}{(1-p)\mathcal{P}_G+(1-q)\mathcal{P}_B}  &\text{ if } x = 0, y = B
\end{cases} \nonumber
\end{align}
\label{regj2}
\end{theorem}
\begin{proof}
We first note that using the total probability theorem we have:
\begin{align}
\mathds{P}(X_i(n)=x) = 
\begin{cases}
p\mathcal{P}_G+q\mathcal{P}_B &\text{ if } x = 1 \\
(1-p)\mathcal{P}_G+(1-q)\mathcal{P}_B &\text{ if } x = 0 
\end{cases}
\end{align}
We now use the Bayes theorem and write:
\begin{align}
&\mathds{P}(\mathcal{C}(n)=y|X_i(n)=x) = \nonumber \\
&\qquad \cfrac{\mathds{P}(\mathcal{C}(n)=y)}{\mathds{P}(X_i(n)=x)}~\mathds{P}(X_i(n)=x|\mathcal{C}(n)=y) 
\end{align}
By simple substitution in the previous expression we obtain:
\begin{align}
&\mathds{P}(\mathcal{C}(n)=y|X_i(n)=x) = \\
&\begin{cases}
\cfrac{p\mathcal{P}_G}{p\mathcal{P}_G+q\mathcal{P}_B} &\text{ if } x = 1, y = G\\
\cfrac{(1-p)\mathcal{P}_G}{(1-p)\mathcal{P}_G+(1-q)\mathcal{P}_B}  &\text{ if } x = 0, y = G\\
\cfrac{q\mathcal{P}_B}{p\mathcal{P}_G+q\mathcal{P}_B}  &\text{ if }  x = 1, y = B\\
\cfrac{(1-q)\mathcal{P}_B}{(1-p)\mathcal{P}_G+(1-q)\mathcal{P}_B}  &\text{ if } x = 0, y = B
\end{cases} \nonumber
\end{align}
\end{proof}

\section{Proof of Theorem 1}

Let us first define $\mathcal{E}_i(t)$ as the cumulative number of transmitted packets from the sender that were erased at user $i$ until time $t$. It is easy to infer that the reception completion event at time $t=C_i(S)$ of a user $i$ will occur when it receives an instantly decodable packet in the $C_i(S)$-th recovery transmission from the sender. Consequently, $\forall~t<=\mathcal{C}_i(S)-1$, the transmission at time $t$ following the schedule $S$ can be one of the following options:
\begin{itemize}
\item The packet can be erased at user $i$ $\Rightarrow$ The transmission will increase $\mathcal{E}_i(t)$ $\left(\mbox{i.e. } \mathcal{E}_i(t)=\mathcal{E}_i(t-1)+1\right)$.
\item The packet can be successfully received by the user $\Rightarrow$ Two cases can occur types:
\begin{itemize}
\item The packet is instantly decodable for user $i$. Note that user $i$ needs to receive $|\mathcal{W}(0)|-1$ of those packets until time $t=C_i(S)-1$ in order to complete its reception by the last missing source packet from the transmitted packet at time $t=C_i(S)$. Consequently, the number of such packets received by user $i$ until time $t=C_i(S)$ is equal to $|\mathcal{W}(0)|$.
\item The packet is either non-innovative or non instantly decodable $\Rightarrow$ This will increase the value of $D_i(S)$ by one each time it occurs until the reception completion for this user.
\end{itemize}
\end{itemize}

Consequently, the number of recovery transmission sent by the sender following schedule $S$ until user $i$ complete its reception of the frame packets (i.e. completion time of user $i$) can be expressed as follows:
\begin{align}
\mathcal{C}_i(S) = |\mathcal{W}_i(0)| + D_i(S) + \mathcal{E}_i(\mathcal{C}_i(S)-1)\;.
\label{cti}
\end{align}

Let $\mathcal{X}_i(t)$ be the number of time instant, from the beginning of the \emph{recovery phase}, until the time $t$, in which the channel was in Good state and let $\mathcal{Y}_i(t)$ be number in which it was in the bad state. Using the limit distribution of the Markov chain, we can write:
\begin{align}
\mathcal{X}_i(t) \approx t\mathcal{P}_{G_i^f} \\
\mathcal{Y}_i(t) \approx t\mathcal{P}_{B_i^f}
\end{align}

$\mathcal{E}_i^g(t)$ and $\mathcal{E}_i^b(t)$ be the number of erased transmission in the Good and Bad state respectively from the beginning of the \emph{recovery phase}, until the time $t$. Let Using the law of large number in each of the states of the Markov chain, we have:
\begin{align}
\mathcal{E}_i^g(t) \approx \mathcal{X}_i(t)p_i^f \approx t\mathcal{P}_{G_i^f}p_i^f \\
\mathcal{E}_i^b(t) \approx \mathcal{Y}_i(t)q_i^f \approx t\mathcal{P}_{B_i^f}q_i^f
\end{align}

For large enough frame size $N$, the completion time $C_i(S)$ would also be large enough and thus $\mathcal{E}_i(\mathcal{C}_i(S)-1)$ can be approximated using the law of large numbers as follows:
\begin{align}
\mathcal{E}_i(\mathcal{C}_i(S)-1) = \mathcal{E}_i(t)^g + \mathcal{E}_i(t)^b \approx \alpha_i(\mathcal{C}_i(S)-1),
\end{align}
where:
\begin{align}
\alpha_i = \cfrac{g_i^fp_i^f+q_i^fb_i^f}{g_i^f+b_i^f}
\end{align}

Substituting the previous expression in \eref{cti} and re-arranging the terms, the completion time for user $i$ can be finally expressed as:
\begin{align}
\mathcal{C}_i(S) \approx \cfrac{|\mathcal{W}_i(0)| + D_i(S) - \alpha_i}{1-\alpha_i}.
\end{align}

Thus, the expression for the overall completion time can be expressed as:
\begin{align}
\mathcal{C}(S) \approx \max_{i\in\mathcal{M}}\left\{\frac{\left|\mathcal{W}_i(0)\right| + D_i(S)-\alpha_i}{1-\alpha_i}\right\}
\end{align}

\section{Proof of Theorem 3}

We first proof that the algorithm, as stated in the original paper, will poorly perform in our system. Let $L$ be the number of particle and $T$ the number of iterations. Assume there is a user $i$ who is missing all the packets (i.e. $\mathcal{H}_i = \varnothing$). Further assume that all users expect of user $i$ received all their packets. Therefore, the only packet combination that can reduce the Wants set of user $i$, is a packet combination where only one packet is included. We will refer to such packet combination as sparse packet combination. For a random initialisation of one of the particle, the probability that a particle will be sparse is $N\left(\cfrac{1}{2}\right)^N$. Thus, the probability that at least one of the $L$ particle is spare is:
\begin{align}
1-\left(1-N\left(\cfrac{1}{2}\right)^N\right)^L
\end{align}

For a large number of packets $N$, with high probability none of the initial value of the $L$ particles will be sparse. For a non-sparse particle, the merit function will be $0$ since no user will be targeted. As a consequence, the update of the particles will be random since almost all direction will result in a non-sparse particle and thus a $0$ merit. Therefore the second iteration of the algorithm can be seen as another initialization of the $L$ particle. The probability to move one particle in a spare configuration after the $T$ iterations of the algorithm is:
\begin{align}
1-\left(1-N\left(\cfrac{1}{2}\right)^N\right)^{L+T}
\end{align}

For a small number of iterations $T$, with high probability, the algorithm will end with a non sparse particles and therefore no update will be made in the system. This process will result in a very poor performance of the overall system. By setting the number of particle equal to the number of packets $L=N$ and using a sparse initialisation of each particle different from the other particles, we can guarantee a decrease in the merit function each time the algorithm is run. Therefore, at each time instant, unless the packet is erased, we can ensure a reduction by at least one packet from the Wants set of user of interest. This conclude to the overall convergence of the system independently of the number of iterations $T$ of the algorithm.

\section{Proof of Theorem 4}

To compute the probability to loose a transmission $e_i(t)$ or the loose a feedback $f_i(t)$ at time instant $t$, lets consider the channel defined in \fref{fig:GEC}. For notation simplicity, we will not consider the superscripts in this theorem unless it is necessary to specify the forward and backward channels. We first compute the following probability:
\begin{align}
\mathds{P}(X_i(n) = 1 |X_i(n^0) = x),
\end{align}
where $X_i(n)$ is a random variable that take the value $1$ if the transmission at time $n$ is erased and $0$ otherwise and $n \geq n_0$. Using the total probability theorem, we write the previous probability as:
\begin{align}
&\mathds{P}(X_i(n)=1|X_i(n^0)=x) =  \nonumber \\
&\mathds{P}(X_i(n)=1|\mathcal{C}(n)=G,X_i(n^0)=x) \nonumber \\
& \hspace{2cm} \times \mathds{P}(\mathcal{C}(n)=G|X_i(n^0)=x) +\nonumber \\
& \mathds{P}(X_i(n)=1|\mathcal{C}(n)=B,X_i(n^0)=x) \nonumber \\
& \hspace{2cm} \times \mathds{P}(\mathcal{C}(n)=B|X_i(n^0)=x) 
\end{align}
By definition of the Markov chain we have:
\begin{align}
&\mathds{P}(X_i(n)=1|\mathcal{C}(n)=G,X_i(n^0)=x) = \nonumber \\ 
& \qquad \mathds{P}(X_i(n)=1|\mathcal{C}(n)=G) = p \\
&\mathds{P}(X_i(n)=1|\mathcal{C}(n)=B,X_i(n^0)=x) = \nonumber \\
& \qquad \qquad \mathds{P}(X_i(n)=1|\mathcal{C}(n)=B) = q
\end{align}
Using the total probability theorem, we can write first term of the previous expression as:
\begin{align}
&\mathds{P}(\mathcal{C}(n)=G|X_i(n^0)=x) = \nonumber \\
&\mathds{P}(\mathcal{C}(n)=G|\mathcal{C}(n^0)=G,X_i(n^0)=x) \nonumber \\
& \hspace{2cm} \times \mathds{P}(\mathcal{C}(n^0)=G|X_i(n^0)=x)  + \nonumber \\
&\mathds{P}(\mathcal{C}(n)=G|\mathcal{C}(n^0)=B,X_i(n^0)=x) \nonumber \\
& \hspace{2cm} \times \mathds{P}(\mathcal{C}(n^0)=B|X_i(n^0)=x) 
\end{align}
We can further reduce the previous expressions:
\begin{align}
&\mathds{P}(\mathcal{C}(n)=G|X_i(n^0)=x) = \nonumber \\
&\mathds{P}(\mathcal{C}(n)=G|\mathcal{C}(n^0)=G)\mathds{P}(\mathcal{C}(n^0)=G|X_i(n^0)=x)  + \nonumber \\
&\mathds{P}(\mathcal{C}(n)=G|\mathcal{C}(n^0)=B)\mathds{P}(\mathcal{C}(n^0)=B|X_i(n^0)=x) 
\end{align}
Similarly, we apply the same development to the second term:
\begin{align}
&\mathds{P}(\mathcal{C}(n)=B|X_i(n^0)=x) = \nonumber \\
&\mathds{P}(\mathcal{C}(n)=B|\mathcal{C}(n^0)=G)\mathds{P}(\mathcal{C}(n^0)=G|X_i(n^0)=x) + \nonumber \\ 
&\mathds{P}(\mathcal{C}(n)=B|\mathcal{C}(n^0)=B)\mathds{P}(\mathcal{C}(n^0)=B|X_i(n^0)=x)
\end{align}

Using Theorem \ref{regj} and Theorem \ref{regj2}, we can express the previous probability according to the value of $x$:
\begin{align}
&\mathds{P}(X_i(n)=1|X_i(n^0)=x) =   \\
& \begin{cases}
&\cfrac{p\mathcal{P}_G}{p\mathcal{P}_G+q\mathcal{P}_B}(p+(q-p)b\sum\limits_{i=0}^{n-n^0-1} \mu^{i})  \nonumber \\ 
& + \cfrac{q\mathcal{P}_B}{p\mathcal{P}_G+q\mathcal{P}_B}(q+(p-q)g\sum\limits_{i=0}^{n-n^0-1} \mu^{i}) \nonumber \\ 
& \qquad \text{ if } x = 1 \nonumber \\
&\cfrac{(1-p)\mathcal{P}_G}{(1-p)\mathcal{P}_G+(1-q)\mathcal{P}_B} (p+(q-p)b\sum\limits_{i=0}^{n-n^0-1} \mu^{i})  \nonumber \\ 
& + \cfrac{(1-q)\mathcal{P}_B}{(1-p)\mathcal{P}_G+(1-q)\mathcal{P}_B}(q+(p-q)g\sum\limits_{i=0}^{n-n^0-1} \mu^{i})   \nonumber \\ 
& \qquad  \text{ if } x = 0
\end{cases}
\end{align}

We now apply our framework to compute the probability to loose a transmission and to loose the feedback. At time $t_i^0$, packet $j^0$ was attempted to user $i$. Therefore, the probability to loose the transmission can be seen as:

\begin{align}
&e_i(t) = \mathds{P}(X_i(n) = 1 |X_i(t_i^0) = x)   = \nonumber \\
&\begin{cases}
\mathds{P}(X_i(n) = 1 |X_i(t_i^0) = 1) \text{ if } f_{ij^0} = 1 \\
\mathds{P}(X_i(n) = 1 |X_i(t_i^0) = 0) \text{ if } f_{ij^0} = 0
\end{cases}
\end{align}

Using the expression derived above, we can write this probability as:
\begin{align}
&e_i = \\
&\begin{cases}
&\cfrac{p_i^fg_i^f}{p_i^fg_i^f+q_i^fb_i^f}(p_i^f+(q_i^f-p_i^f)b_i^f\sum\limits_{i=0}^{t-t_i^{(0)}-1} \mu^{i})  \nonumber \\ 
& + \cfrac{q_i^fb_i^f}{p_i^fg_i^f+q_i^fb_i^f}(q_i^f+(p_i^f-q_i^f)g\sum\limits_{i=0}^{t-n^0-1} \mu^{i}) \nonumber \\ 
& \qquad \text{ if } f_{ij^0} = 1 \nonumber \\
&\cfrac{(1-p_i^f)g_i^f}{(1-p_i^f)g_i^f+(1-q_i^f)b_i^f}  \nonumber \\
& \hspace{2cm} \times  (p_i^f+(q_i^f-p_i^f)b_i^f\sum\limits_{i=0}^{t-t_i^{(0)}-1} \mu^{i})  \nonumber \\ 
& + \cfrac{(1-q_i^f)b_i^f}{(1-p_i^f)g_i^f+(1-q_i^f)b_i^f} \nonumber \\
& \hspace{2cm} \times  (q_i^f+(p_i^f-q_i^f)g_i^f\sum\limits_{i=0}^{t-t_i^{(0)}-1} \mu^{i}) \nonumber \\ 
& \qquad  \text{ if } f_{ij^0} = 0
\end{cases}
\end{align}

At time $t_i^*$ the feedback was successfully received from user $i$. Thus, the probabilities $f_i(t)$ of loosing a feedback from user $i$ at time $t>t_i^*$ can be expressed as:
\begin{align}
&f_i = \mathds{P}(X_i(n) = 1 |X_i(t_i^*) = 0) \nonumber \\
& = \cfrac{(1-p_i^b)g_i^b}{(1-p_i^b)g_i^b+(1-q_i^b)b_i^b} (p_i^b+(q_i^b-p_i^b)b_i^b\sum_{i=0}^{t-t_i^*-1} \psi^{i})  \nonumber \\ 
& + \cfrac{(1-q_i^b)b_i^b}{(1-p_i^b)g_i^b+(1-q_i^b)b_i^b}(q_i^b+(p_i^b-q_i^b)g_i^b\sum_{i=0}^{t-t_i^*-1} \psi^{i})   
\end{align}

Note that we can obtain the expressions derived in \cite{refjournal} by setting $p=0$ and $q=1$:
\begin{align}
&e_i =
\begin{cases}
1-g_i^f\sum\limits_{i=0}^{n-t_i^0-1} \mu^{i}&\text{ if } f_{ij^0} = 1 \\
b_i^f\sum\limits_{i=0}^{n-t_i^0-1} \mu^{i} &\text{ if } f_{ij^0} = 0
\end{cases} \\
&f_i =  b_i^b\sum\limits_{i=0}^{n-t_i^*-1} \psi^{i}
\end{align}

\section{Expressions of Innovative and Finish Probabilities}

Let $\mathcal{K}_{ij}$ be the set of indexes of the frames in which packet $j$ was attempted to user $i$ since the last time the BS received feedback from this user, excluding the current frame. Define $\mathcal{U}_i^d(n)$ as the following:
\begin{align}
\mathcal{U}_i^d(n) = \bigcup_{j \in \mathcal{W}_i(n \times T_f)}\lambda_{ij}(n),~ \forall~ n \in \mathds{N}^+.
\end{align}

Given these definitions, the probability $p_{i,n}(j,t)$ that packet $j$ is innovative for user $i$ can be expressed as:
\begin{align}
&p_{i,n}(j,t) = \langle \prod_{k \in \lambda_{ij}(n^+(t))} e_i(k) \rangle \nonumber\\
& \times \langle \prod_{k \in \mathcal{K}_{ij}} \rangle \left\{ \left( \prod_{s \in \mathcal{U}_i^d(k)} e_i(s)  +  \prod_{s \in \lambda_{ij}(k)} e_i(s)  \right. \right. \nonumber\\
& \left. {} \left. {} \times (1- \prod_{s \in \mathcal{U}_i^d(k) \setminus \lambda_{ij}(k)} e_i(s) )f_i(u_i(k)) \right) \right. \nonumber\\
& \left. {} \times \left( \prod_{s \in \mathcal{U}_i^d(k)} e_i(s)  + (1- \prod_{s \in \mathcal{U}_i^d(k)} e_i(s) )f_i(u_i(k)) \right)^{-1} \right\} 
\end{align}
The probability $p_{i,f}(t)$ that user $i$ successfully received all his primary packets but $\mathcal{W}_i(t) \neq \varnothing$, at time $t$ is the following:
\begin{align}
p_{i,f}(t)  = \prod_{j \in \mathcal{W}_i(t)} \left( 1- p_{i,n}(j,t) \right),
\end{align}
where $u_i(n)=n*T_f-T_u+T_{u_i}$ and
\begin{align}
\langle \prod_{x \in X} (.) \rangle = 
\begin{cases}
\prod\limits_{x \in X} (.) &\text{if } X \neq \varnothing \\
1  &\text{if } X = \varnothing .
\end{cases} 
\end{align} 
\begin{proof}
The proof can be found in \cite{refjournal} Appendix D.
\end{proof}

\section{Proof of Theorem 5}

Let $\mathds{M}(t)$ be the event that $\max_{i\in\mathcal{M}}\left\{\mathcal{C}_i(t)\right\} > \max_{i\in\mathcal{M}}\left\{\mathcal{C}_i(t-1)\right\}$ at time after a transmission $\kappa(t)$ at time $t$. The probability of this event can be expressed as:
\begin{align}
\mathds{P}(\mathds{M}(t))
&=  \mathds{P}\left(\max_{i\in\mathcal{M}}\left\{\mathcal{C}_i(t)\right\} > \max_{i\in\mathcal{M}}\left\{\mathcal{C}_i(t-1)\right\}\right) \nonumber \\
&= 1 - \mathds{P}\left(\max_{i\in\mathcal{M}}\left\{\mathcal{C}_i(t)\right\} = \max_{i\in\mathcal{M}}\left\{\mathcal{C}_i(t-1)\right\}\right).
\end{align}

Users $j \in \mathcal{M} \setminus \mathcal{P}(t)$ are unable to increase $\max_{i\in\mathcal{M}}\left\{\mathcal{C}_i(t)\right\}$ compared to $\max_{i\in\mathcal{M}}\left\{\mathcal{C}_i(t-1)\right\}$ with probability 1, even if they experience a decoding delay. This is true since the set $\mathcal{P}(t)$ is constructed such that it contains all users that have non-zero probabilities of increasing the completion time. According the definition of $\mathcal{C}_i(t)$ in \eref{eq:Ct}, $\forall~i\in\mathcal{M}$, users $i \in \mathcal{P}(t)$ will not increase $\max_{i\in\mathcal{M}}\left\{\mathcal{C}_i(t)\right\}$ after the transmission $\kappa(t)$ only if they do not experience a decoding delay increment in this transmission. Consequently, we get:
\begin{align}
& \mathds{P}\left(\max_{i\in\mathcal{M}}\left\{\mathcal{C}_i(t)\right\} = \max_{i\in\mathcal{M}}\left\{\mathcal{C}_i(t-1)\right\}\right) \nonumber \\
& \qquad = \mathds{P}\left(\max_{i\in\mathcal{P}(t)}\left\{\mathcal{C}_i(t)\right\} = \max_{i\in\mathcal{M}}\left\{\mathcal{C}_i(t-1)\right\}\right) \nonumber \\
& \qquad = \mathds{P} \left(D_i(t)-D_i(t-1)=0, \forall~i\in\mathcal{P}(t)\right) \nonumber \\
& \qquad = \prod_{i \in \mathcal{P}(t)} \mathds{P} \left(D_i(t)-D_i(t-1)=0\right).
\end{align}

According to the analysis done in \cite{refjournal}, the probability of the decoding delay increase for user $i$ is given by the following theorem:
\begin{theorem}
The probability that user $i$ does not experience a decoding delay at time $t$, after the transmission $\kappa$ is:
\begin{align}
&\mathds{P}(d_i(\kappa,t) = 0) \nonumber \\
&= 
\begin{cases}
e_i(t)&  i \in (\widehat{\tau} \cap \overline{F}) \\
e_i(t)+p_{i,f}(t)-e_i(t)p_{i,f}(t) &i \in (\widehat{\tau} \cap F) \\
1&  i \in (\tau \cap \overline{U}) \\
e_i(t)+p_{i,n}(\kappa_i,t)&  i \in (\tau \cap (U \setminus F)) \\
\qquad -e_i(t)p_{i,n}(\kappa_i,t)  \\
e_i(t)+(1-e_i(t))(p_{i,n}(\kappa_i,t)& i \in (\tau \cap F)  \\
\qquad +p_{i,f}(t)) ,
\end{cases}
\end{align}
where $\widehat{\tau}$ is set of users not targeted and having non-empty Wants sets (i.e. $\widehat{\tau}= M_w \setminus \tau$), $\kappa_i$ is the intended packet for user $i$ in the transmission $\kappa$, $F$ is the set of users having all their remaining packets in an uncertain state and $U$ is the set of users having the intended packet for them in an uncertain state. The notation $\overline{X}$ refers to the set complementary to the set $X$.
\end{theorem}
\begin{proof}
The proof can be found in \cite{refjournal} Appendix A.
\end{proof}

In other words, the completion time does not increase only if all the users having the completion time so far do not experience a decoding delay increase in the next transmission. Using the expression of the decoding delay increase, the probability of event $\mathds{M}(t)$ to occur can be expressed as follows:
\begin{align}
\mathds{P}&(\mathds{M}(t)) = 1 - \prod_{i \in \mathcal{P}(t)} \mathds{P} (d_i(\kappa,t) = 0) \nonumber \\
& \hspace{1cm} = 1 -  \prod_{i \in (\mathcal{P}(t) \cap \widehat{\tau} \cap \overline{F}) } e_i(t) \nonumber \\
&\prod_{i \in (\mathcal{P}(t) \cap \widehat{\tau} \cap F) } e_i(t)+p_{i,f}(t)-e_i(t)p_{i,f}(t) \nonumber \\
&\prod_{i \in (\mathcal{P}(t) \cap \tau \cap (U \setminus F)) } e_i(t)+p_{i,n}(\kappa_i,t)-e_i(t)p_{i,n}(\kappa_i,t) \nonumber \\
&\prod_{i \in (\mathcal{P}(t) \cap \tau \cap F) } e_i(t)+(1-e_i(t))(p_{i,n}(\kappa_i,t)+p_{i,f}(t)).
\end{align}

From the expressions of the completion time increment, we can express the minimum completion time problem as a maximum weight clique problem, such that:
\begin{align}
\kappa^{*}&(t) = \underset{\kappa(t) \in \mathcal{G}}{\text{argmin}}\left\{ \mathds{P}(\mathds{M}(t))\right\} \nonumber \\
&= \underset{\kappa(t) \in \mathcal{G}}{\text{argmin}} \left\{ 1 -  \prod_{i \in (\mathcal{P}(t) \cap \widehat{\tau} \cap \overline{F}) } e_i(t) \right. \nonumber \\
& \left. {} \prod_{i \in (\mathcal{P}(t) \cap \widehat{\tau} \cap F) } e_i(t)+p_{i,f}(t)-e_i(t)p_{i,f}(t) \right. \nonumber \\
& \left. {} \prod_{i \in (\mathcal{P}(t) \cap \tau \cap (U \setminus F)) } e_i(t)+p_{i,n}(\kappa_i,t)-e_i(t)p_{i,n}(\kappa_i,t) \right. \nonumber \\
& \left. {} \prod_{i \in (\mathcal{P}(t) \cap \tau \cap F) } e_i(t)+(1-e_i(t))(p_{i,n}(\kappa_i,t)+p_{i,f}(t)) \right\}   \nonumber \\
&= \underset{\kappa(t) \in \mathcal{G}}{\text{argmax}} \left\{\prod_{i \in (\mathcal{P}(t) \cap \widehat{\tau} \cap \overline{F}) } e_i(t) \right. \nonumber \\
& \left. {} \prod_{i \in (\mathcal{P}(t) \cap \widehat{\tau} \cap F) } e_i(t)+p_{i,f}(t)-e_i(t)p_{i,f}(t) \right. \nonumber \\
& \left. {} \prod_{i \in (\mathcal{P}(t) \cap \tau \cap (U \setminus F)) } e_i(t)+p_{i,n}(\kappa_i,t)-e_i(t)p_{i,n}(\kappa_i,t) \right. \nonumber \\
& \left. {} \prod_{i \in (\mathcal{P}(t) \cap \tau \cap F) } e_i(t)+(1-e_i(t))(p_{i,n}(\kappa_i,t)+p_{i,f}(t)) \right\}   .
\end{align}
Since the function $log(.)$ is an increasing function, then the problem can be expressed as:
\begin{align}
&\kappa^{*}(t) = \underset{\kappa(t) \in \mathcal{G}}{\text{argmin }} \text{log}\left\{ \mathds{P}(\mathds{M}(t))\right\} \nonumber \\
&= \underset{\kappa(t) \in \mathcal{G}}{\text{argmax}} \left\{ \sum_{i \in (\mathcal{P}(t) \cap \widehat{\tau} \cap \overline{F}) } \text{log}( e_i(t)) \right. \nonumber \\
& \left. {} +\sum_{i \in (\mathcal{P}(t) \cap \widehat{\tau} \cap F) } \text{log}( e_i(t) +p_{i,f}(t)-e_i(t)p_{i,f}(t)) \right. \nonumber \\
& \left. {} +\sum_{i \in (\mathcal{P}(t) \cap \tau \cap (U \setminus F)) } \text{log}( e_i(t)+p_{i,n}(\kappa_i,t)-e_i(t)p_{i,n}(\kappa_i,t)) \right. \nonumber \\
& \left. {} +\sum_{i \in (\mathcal{P}(t) \cap \tau \cap F) } \text{log}( e_i(t)+(1-e_i(t))(p_{i,n}(\kappa_i,t)+p_{i,f}(t))) \right\}   .
\end{align}

If user $i$ does not have all its wanted packets in an uncertain state, then the probability than he finished receiving all its wanted packet is $0$. Thus $p_{i,f}(t)=0,~\forall~ i \in \overline{F}$. Therefore, we have:
\begin{align}
&\sum_{i \in (\mathcal{P}(t) \cap \widehat{\tau} \cap \overline{F}) } \text{log}( e_i(t)) \nonumber\\ 
&= \sum_{i \in (\mathcal{P}(t) \cap \widehat{\tau} \cap \overline{F}) } \text{log}( e_i(t)+p_{i,f}(t)-e_i(t)p_{i,f}(t)).
\label{q1}
\end{align}

Using \eref{q1}, the expression below can be simplified as:
\begin{align}
& \sum_{i \in (\mathcal{P}(t) \cap \widehat{\tau} \cap \overline{F}) } \text{log}( e_i(t)) \nonumber \\
& +\sum_{i \in (\mathcal{P}(t) \cap \widehat{\tau} \cap F) } \text{log}( e_i(t)+p_{i,f}(t)-e_i(t)p_{i,f}(t))  \nonumber \\
& = \sum_{i \in (\mathcal{P}(t) \cap \widehat{\tau} )} \text{log}( e_i(t)+p_{i,f}(t)-e_i(t)p_{i,f}(t)) .
\end{align}

It is clear that $(U \setminus F) \subseteq \overline{F}$, then $p_{i,f}(t)=0,~\forall~ i \in (U \setminus F)$. We then obtain:
\begin{align}
& \sum_{i \in (\mathcal{P}(t) \cap \tau \cap (U \setminus F)) } \text{log}( e_i(t)+p_{i,n}(\kappa_i,t)-e_i(t)p_{i,n}(\kappa_i,t)) \nonumber \\
&= \sum_{i \in (\mathcal{P}(t) \cap \tau \cap (U \setminus F)) }\text{log}( e_i(t) \nonumber \\
&  +(1-e_i(t))(p_{i,n}(\kappa_i,t)+p_{i,f}(t))).
\label{q2}
\end{align}

Therefore, using \eref{q2}, we can simplify the below expression:
\begin{align}
& \sum_{i \in (\mathcal{P}(t) \cap \tau \cap (U \setminus F)) } \text{log}( e_i(t)+p_{i,n}(\kappa_i,t)-e_i(t)p_{i,n}(\kappa_i,t))  \nonumber \\
& +\sum_{i \in (\mathcal{P}(t) \cap \tau \cap F) } \text{log}( e_i(t)+(1-e_i(t))(p_{i,n}(\kappa_i,t)+p_{i,f}(t)))   \nonumber \\
&= \sum_{i \in (\mathcal{P}(t) \cap \tau \cap U) } \text{log}( e_i(t)+(1-e_i(t))(p_{i,n}(\kappa_i,t)+p_{i,f}(t))).
\end{align}

Given the above simplifications, the maximum weight clique problem can be written as follows:
\begin{align}
&\kappa^{*}(t) = \underset{\kappa(t) \in \mathcal{G}}{\text{argmin }} \text{log}\left\{ \mathds{P}(\mathds{M}(t))\right\} \nonumber \\
&  = \underset{\kappa(t) \in \mathcal{G}}{\text{argmax}} \left\{ \sum_{i \in (\mathcal{P}(t) \cap \widehat{\tau} ) } \text{log}( e_i(t)+p_{i,f}(t)-e_i(t)p_{i,f}(t)) \right. \nonumber \\
& \left. {} +\sum_{i \in (\mathcal{P}(t) \cap \tau \cap U) } \text{log}( e_i(t) +(1-e_i(t))(p_{i,n}(\kappa_i,t)+p_{i,f}(t))) \right\}  \nonumber \\  
&  = \underset{\kappa(t) \in \mathcal{G}}{\text{argmin}} \left\{ \sum_{i \in (\mathcal{P}(t) \cap \tau ) } \text{log}( e_i(t)+p_{i,f}(t)-e_i(t)p_{i,f}(t)) \right. \nonumber \\
& \left. {} -\sum_{i \in (\mathcal{P}(t) \cap \tau \cap U) } \text{log}( e_i(t)+(1-e_i(t))(p_{i,n}(\kappa_i,t)+p_{i,f}(t))) \right\}  .
\end{align}

Note that if the targeted packet $\kappa_i$ of user $i$ is not an uncertain packet (i.e. $i \in \overline{U}$), then the packet is certainly innovative. Since that this user have at least one certain wanted packet then he surely still needs packets. In other words, we have $i \in \overline{U} \Rightarrow p_{i,n}(\kappa_i,t)=1$ and  $p_{i,f}(t)=0$. We write the following expression as:
\begin{align}
&\sum_{i \in (\mathcal{P}(t) \cap \tau) } \text{log}( e_i(t)+(1-e_i(t))(p_{i,n}(\kappa_i,t)+p_{i,f}(t))) \nonumber \\
&= \sum_{i \in (\mathcal{P}(t) \cap \tau \cap U} \text{log}( e_i(t)+(1-e_i(t))(p_{i,n}(\kappa_i,t)+p_{i,f}(t))) \nonumber \\
&+ \sum_{i \in (\mathcal{P}(t) \cap \tau \cap \overline{U}} \text{log}( e_i(t)+(1-e_i(t))(p_{i,n}(\kappa_i,t)+p_{i,f}(t))) \nonumber \\
&= \sum_{i \in (\mathcal{P}(t) \cap \tau \cap U} \text{log}( e_i(t)+(1-e_i(t))(p_{i,n}(\kappa_i,t)+p_{i,f}(t))).
\end{align}

Giving all the above simplifications, we now express the maximum weight clique problem as:
\begin{align*}
&\kappa^{*}(t) = \underset{\kappa(t) \in \mathcal{G}}{\text{argmin }} \text{log}\left\{ \mathds{P}(\mathds{M}(t))\right\} \nonumber \\ 
&= \underset{\kappa(t) \in \mathcal{G}}{\text{argmin}} \left\{ \sum_{i \in (\mathcal{P}(t) \cap \tau ) } \text{log}( e_i(t)+p_{i,f}(t)-e_i(t)p_{i,f}(t)) \right. \nonumber \\
& \left. {} -\sum_{i \in (\mathcal{P}(t) \cap \tau ) } \text{log}( e_i(t)+(1-e_i(t))(p_{i,n}(\kappa_i,t)+p_{i,f}(t))) \right\}  \nonumber \\  
&= \underset{\kappa(t) \in \mathcal{G}}{\text{argmin}} \nonumber \\ 
& \left\{ \sum_{i \in (\mathcal{P}(t) \cap \tau ) } \text{log}\left(\cfrac{ e_i(t)+p_{i,f}(t)-e_i(t)p_{i,f}(t)}{e_i(t)+(1-e_i(t))(p_{i,n}(\kappa_i,t)+p_{i,f}(t))} \right) \right\} \nonumber \\  
\end{align*}
\begin{align}
&= \underset{\kappa(t) \in \mathcal{G}}{\text{argmax}} \nonumber \\  
&\left\{ \sum_{i \in (\mathcal{P}(t) \cap \tau ) } \text{log}\left(\cfrac{e_i(t)+(1-e_i(t))(p_{i,n}(\kappa_i,t)+p_{i,f}(t))}{ e_i(t)+p_{i,f}(t)-e_i(t)p_{i,f}(t)}\right) \right\} \nonumber \\ 
&= \underset{\kappa(t) \in \mathcal{G}}{\text{argmax}} \nonumber \\  
&\left\{ \sum_{i \in (\mathcal{P}(t) \cap \tau ) } \text{log}\left(1 + \cfrac{(1-e_i(t))p_{i,n}(\kappa_i,t)}{ e_i(t)+p_{i,f}(t)-e_i(t)p_{i,f}(t)}\right) \right\} \nonumber \\ 
&= \underset{\kappa(t) \in \mathcal{G}}{\text{argmax}} \left\{ \sum_{i \in (\mathcal{P}(t) \cap \tau ) } \text{log}\left( 1 + \cfrac{p_{i,n}(\kappa_i,t)}{ \cfrac{e_i(t)}{1-e_i(t)}+p_{i,f}(t)}\right) \right\}.
\end{align}

In other words, the transmission $\kappa(t)$ than can satisfy the critical criterion can be selected using a maximum weight clique problem in which the weight of each vertex $v_{ij}$ in $\mathcal{P}(t)$ can be expressed as:
\begin{align}
w_{ij}^* = \text{log}\left( 1 + \cfrac{p_{i,n}(j,t)}{ \cfrac{e_i(t)}{1-e_i(t)}+p_{i,f}(t)}\right).
\end{align}

\bibliographystyle{IEEEtran}
\bibliography{references}

\end{document}